\documentclass[letterpaper, 10 pt, conference]{ieeeconf}  

\IEEEoverridecommandlockouts                              

\usepackage[pdftex, pdfstartview={FitV}, pdfpagelayout={TwoColumnLeft},bookmarksopen=true,plainpages = false, colorlinks=true, linkcolor=black, citecolor = black, urlcolor = black,filecolor=black , pagebackref=false,hypertexnames=false, plainpages=false, pdfpagelabels ]{hyperref}

\usepackage[authormarkuptext=name,addedmarkup=bf,authormarkupposition=right]{changes}
\definechangesauthor[name={B.~L.}, color={blue}]{bl}

\makeatletter
\def\set@curr@file#1{%
  \begingroup
    \escapechar\m@ne
    \xdef\@curr@file{\expandafter\string\csname #1\endcsname}%
  \endgroup
}
\def\quote@name#1{"\quote@@name#1\@gobble""}
\def\quote@@name#1"{#1\quote@@name}
\def\unquote@name#1{\quote@@name#1\@gobble"}
\makeatother
\usepackage{graphicx}

\usepackage{epstopdf}
\usepackage{amssymb,amsmath,mathrsfs}

\usepackage{amsthm}
\usepackage[font=footnotesize]{subcaption}

\usepackage[font=footnotesize]{caption}
\usepackage{xcolor}
\usepackage{units}
\usepackage{algorithm}
\usepackage[noend]{algpseudocode}
\usepackage{balance}
\usepackage[sort,compress,noadjust]{cite}
\usepackage{tabularx}
\usepackage{nameref}

\let\labelindent\relax
\usepackage{enumitem}

\theoremstyle{plain}
\newtheorem{theorem}{Theorem}
\theoremstyle{plain}

\theoremstyle{plain}
\newtheorem{lemma}{Lemma}
\theoremstyle{plain}

\theoremstyle{definition}
\newtheorem{assumption}{Assumption}
\theoremstyle{definition}
\newtheorem{definition}{Definition}
\theoremstyle{remark}
\newtheorem{remark}{Remark}
\theoremstyle{remark}
\newtheorem{example}{Example}

\usepackage[hang,flushmargin]{footmisc}

\usepackage[capitalize]{cleveref}
\crefformat{equation}{(#2#1#3)}
\Crefformat{equation}{Equation~(#2#1#3)}
\Crefname{equation}{Equation}{Eqs.}

\usepackage{accents}
\newcommand{\ubar}[1]{\underaccent{\bar}{#1}}

\DeclareMathOperator*{\argmin}{arg\,min}

\title{\LARGE \bf
Robust Adaptive Control Barrier Functions: \\
An Adaptive \& Data-Driven Approach to Safety  \\
(Extended Version)
}

\author{Brett T. Lopez$^{1}$, Jean-Jacques E. Slotine$^{2}$, and Jonathan P. How$^{1}$
\thanks{$^{1}$Aerospace Control Laboratory, Massachusetts Institute of Technology, Cambridge MA, {\tt\small \{btlopez,jhow\}@mit.edu}}
\thanks{$^{2}$Nonlinear Systems Laboratory, Massachusetts Institute of Technology, Cambridge MA, {\tt\small jjs@mit.edu}}
}

\begin{document}

\maketitle
\thispagestyle{empty}
\pagestyle{empty}

\begin{abstract}
A new framework is developed for control of constrained nonlinear systems with structured parametric uncertainties.
Forward invariance of a safe set is achieved through online parameter adaptation and data-driven model estimation.
The new adaptive data-driven safety paradigm is merged with a recent adaptive control algorithm for systems nominally contracting in closed-loop.
This unification is more general than other safety controllers as closed-loop contraction does not require the system be invertible or in a particular form.
Additionally, the approach is less expensive than nonlinear model predictive control as it does not require a full desired trajectory, but rather only a desired terminal state.
The approach is illustrated on the pitch dynamics of an aircraft with uncertain nonlinear aerodynamics.
\end{abstract}


\section{Introduction}
State and actuator constraints are often encountered in real-world systems but systematic feedback controller design remains challenging.
The main difficulty arises from needing to \emph{predict} whether the system will remain in the feasible set when selecting a control input.
Repeatedly solving a constrained finite-horizon optimal control problem (i.e., model predictive control) is one way to ensure feasibility, 
but solving a nonlinear optimization in real-time can be difficult.
Alternatively, one can avoid trajectory optimization entirely by constructing safe invariant sets, i.e., a set of states that guarantee feasibility indefinitely.
Through the development of control barrier functions (CBFs)~\cite{ames2016control}, safe stabilizing controllers can be synthesized by simply solving a quadratic program (QP), and has recently been used in several applications~\cite{ames2019control}.
However, model error can significantly degrade the performance of these controllers to the extent that safety may no longer be guaranteed.
We develop a general framework that guarantees safety through parameter adaptation and online model estimation for uncertain nonlinear systems.

Control barrier functions heavily rely on a model so it is critical to develop methodologies that maintain safety for uncertain systems.
In~\cite{xu2015robustness} the so-called zeroing CBFs were shown to be Input-to-State stable, a property that was used to prove a superset of a safe set is forward invariant.
The size of the superset was characterized in~\cite{kolathaya2018input} by introducing Input-to-State safety.
Stronger safety guarantees can be obtained via robust optimization as demonstrated in~\cite{gurriet2018towards}.
Learning-based methods~\cite{fan2019bayesian,taylor2019learning} have been developed to address the conservatism of the robust strategies, but these can require extensive offline training to substantially improve the model.
Adaptive CBFs (aCBFs)~\cite{taylor2019adaptive} use ideas from adaptive control theory to ensure a safe set is forward invariant with online parameter adaptation.
However, aCBFs have a much more restrictive invariance condition that limits the system to remain in sublevel sets of the safety set, ultimately leading to conservative behavior.
To address conservatism, \cite{taylor2019adaptive} limited parameter adaptation to a region near the boundary of the safe set.
However, the resulting safety controller is not necessarily Lipschitz and can exhibit closed-loop chattering.

The contributions of this work are threefold.
First, robust aCBFs (RaCBFs) are defined and shown to guarantee safety for uncertain nonlinear systems.
When combined with parameter adaptation, RaCBFs ensure forward invariance of a \emph{tightened set} where the degree of tightening can be selected based on the desired conservatism.
RaCBFs are far less conservative than aCBFs and results in a locally Lipschitz safety controller.
Second, RaCBFs are combined with the data-driven method set membership identification \cite{tanaskovic2014adaptive} to safely reduce modeling error and expand the set of allowable states. 
This is the first work to utilize both parameter adaptation and data-driven model estimation within the context of safety-critical control. 
And third, RaCBFs are merged with a recent direct adaptive controller \cite{lopez2019contraction} based on the contraction metric framework \cite{lohmiller1998contraction,manchester2017control}.
Contraction expresses distance-like functions \emph{differentially} rather than explicitly, and the existence of a metric only requires stabilizability -- a far weaker condition than those needed for feedback linearization or backstepping -- making the unification more general than existing methods.
The approach is demonstrated on the pitch dynamics of an aircraft with uncertain nonlinear aerodynamics.
The system is non-invertible, not in strict-feedback form, and has non-polynomial dynamics highlighting the generality of the proposed method.

\section{Problem Formulation \& Preliminaries}
Consider the nonlinear system
\protect\begin{equation}
\dot{x} = f(x) - \Delta(x)^\top \theta + B(x)u,
\label{eq:unc_dyn}
\end{equation}
with unknown parameters $\theta \in \mathbb{R}^p$ and known dynamics $\Delta : \mathbb{R}^n \rightarrow \mathbb{R}^{p\times n} $, state $x \in \mathbb{R}^n$, control input $u \in \mathbb{R}^m$, nominal dynamics $f: \mathbb{R}^n \rightarrow \mathbb{R}^{n}$, and control input matrix $B : \mathbb{R}^n \rightarrow \mathbb{R}^{n\times m}$ with columns $b_i(x)$ $i=1,\dots,m$.
Let $\mathcal{S}^n_+$ be the set of all $n\times n$ symmetric positive definite matrices. 
A smooth Riemannian manifold $\mathcal{M}$ is equipped with a smooth Riemannian metric $M: \mathbb{R}^n\times\mathbb{R} \rightarrow \mathcal{S}^n_+$ that defines an inner product $\left< \cdot, \cdot \right>_x$ on the tangent space $T_x \mathcal{M}$ at every point $x$.
The metric $M(x,t)$ defines local geometric notions such as angles, length, and orthogonality.
The directional derivative of metric $M(x,t)$ along vector $v$ is $\partial_v M = \sum_i\frac{\partial M}{\partial x_i} v_i$.
A parameterized differentiable curve $c: \left[0,1\right] \rightarrow \mathcal{M}$ is a regular if $\frac{\partial c}{\partial s} = c_s \neq 0$ $\forall s \in \left[0~1\right]$.
Let $\Upsilon(p,q)$ be the family of curves connecting $p,~q\in \mathcal{M}$, then a \emph{geodesic} $\gamma: \left[0~1\right] \rightarrow \mathcal{M}$ is the extremum of the energy functional
\begin{equation*}
    \gamma(s) = \underset{c(s)\in\Upsilon(p,q)}{\argmin}~ E(c,t) = \int_0^1 c_s^\top M(c,t) c_s ds,
\end{equation*}
where $E$ is the Riemannian energy. 
If the manifold $\mathcal{M}$ is a complete metric space, such as $\mathbb{R}^n$, $n$-sphere $\mathbb{S}^n$, or any of their respective closed subsets, then a geodesic is guaranteed to exist by the Hopf-Rinow theorem \cite{carmo1992riemannian}.
In the sequel, the time argument in $M(c,t)$ and $E(c,t)$ is dropped for clarity.

\section{Adaptive Safety}
\label{sec:acbf}
\subsection{Background}
First consider the nominal dynamics of \cref{eq:unc_dyn}, i.e., $\Delta(x) = 0$.
Let a closed convex set $\mathcal{C} \subset \mathbb{R}^n$ be a 0-superlevel set of a continuously differentiable function $h: \mathbb{R}^n \rightarrow \mathbb{R}$ where
\begin{equation*}
    \begin{aligned}
    \mathcal{C} & = \left\{ x \in \mathbb{R}^n : h(x) \geq 0  \right\} \\
    \partial \mathcal{C} & = \left\{ x \in \mathbb{R}^n : h(x) = 0  \right\} \\
    \text{Int}\left(\mathcal{C}\right) & = \left\{ x \in \mathbb{R}^n : h(x) > 0  \right\}.
    \end{aligned}
\end{equation*}
If the nominal dynamics are locally Lipschitz, then given an initial condition $x_0$, there exists a maximum time interval $I(x_0) = [t_0,~T)$ such that $x(t)$ is a unique solution on $I(x_0)$. 
The following definitions are largely taken from \cite{ames2016control,ames2019control}.
\begin{definition}
\label{def:fi}
The set $\mathcal{C}$ is \emph{forward invariant} if for every $x_0\in \mathcal{C}$, $x(t) \in \mathcal{C}$ for all $t\in I(x_0)$.
\end{definition}

\begin{definition}
\label{def:safety}
The nominal system is \emph{safe} with respect to set $\mathcal{C}$ if the set $\mathcal{C}$ is forward invariant.
\end{definition}

\begin{definition}
A continuous function $\alpha : \mathbb{R} \rightarrow \mathbb{R}$ is an \emph{extended class $\mathcal{K}_\infty$ function} if it is strictly increasing, $\alpha(0) = 0$, and is defined on the entire real line.
\end{definition}

\begin{definition}
Let $\mathcal{C}$ be a 0-superlevel set for a continuously differentiable function $h:\mathbb{R}^n\rightarrow\mathbb{R}$, then $h$ is a \emph{control barrier function} if there exists an extended class $\mathcal{K}_\infty$ function $\alpha$ such that
\begin{equation}
\label{eq:cbf}
    \underset{u \in \mathcal{U}}{\text{sup}}~ \left[\frac{\partial h}{\partial x}(x)\left(f(x) + B(x) u\right)\right] \geq - \alpha(h(x)).
\end{equation}
\end{definition}

\begin{theorem}
Let $\mathcal{C}\subset\mathbb{R}^n$ be a 0-superlevel set of a continuously differentiable function $h:\mathbb{R}^n\rightarrow\mathbb{R}$, if $h$ is a CBF on $\mathcal{C}$ then any locally Lipschitz continuous controller satisfying \cref{eq:cbf} renders the set $\mathcal{C}$ safe for the nominal system. 
\end{theorem}

\subsection{Adaptive CBFs}
Adaptive CBFs (aCBFs) \cite{taylor2019adaptive} provide a general framework to guarantee safety through parameter adaptation for systems with structured uncertainties.
The notion of safety for uncertain systems must be extended to a family of safe sets $\mathcal{C}_{\theta}$ parameterized by $\theta$. 
More precisely, the family of safe sets are 0-superlevel sets of a continuously differentiable function $h_a: \mathbb{R}^n \times \mathbb{R}^p \rightarrow \mathbb{R}$.
If the uncertain dynamics in \cref{eq:unc_dyn} are locally Lipschitz then the definitions of forward invariance and safety can be directly extended to $\mathcal{C}_\theta$.

\begin{definition}[\cite{taylor2019adaptive}]
\label{def:acbf}
Let $\mathcal{C}_\theta$ be a family of 0-superlevel sets parameterized by $\theta$ for a continuously differentiable function $h_a:\mathbb{R}^n\times \mathbb{R}^p\rightarrow\mathbb{R}$, then $h_a$ is an \emph{adaptive control barrier function} if for all $\theta$
\begin{equation}
\label{eq:acbf}
    \underset{u \in \mathcal{U}}{\text{sup}}~ \left[ \frac{\partial h_a}{\partial x}(x,\theta)\left( f(x) - \Delta(x)^\top \Lambda(x,\theta) + B(x) u\right)\right] \geq 0,
\end{equation}
where $\Lambda(x,\theta) := \theta - \Gamma \left(\frac{\partial h_a}{\partial \theta}(x,\theta)\right)^\top$ and $\Gamma \in \mathcal{S}^p_+$ is a symmetric positive definite matrix
\end{definition}

A controller that satisfies \cref{eq:acbf} can be combined with an adaptation law to render the uncertain systems safe with respect to $\mathcal{C}_{{\theta}}$ \cite{taylor2019adaptive}.
However, \cref{eq:acbf} makes the level sets of $h_a$ forward invariant so it is a much stricter condition than \cref{eq:cbf}.
More precisely, the distance to the boundary of the safe set must monotonically increase, i.e., $\dot{h}_a(x,\theta) \geq 0$ for all time (\cref{fig:acbf}).
This can lead to extremely conservative behavior as the system only operates in a set that is monotonically shrinking.
In \cite{taylor2019adaptive}, a modified aCBF was proposed
\begin{equation}
\begin{aligned}
\label{eq:acbf_m}
    \bar{h}_a(x,\theta) = \begin{cases} 
    \sigma^2 ~ &\text{if} ~ h_a(x,\theta)\geq\sigma \\
    \sigma^2 - (h_a(x,\theta)-\sigma)^2 ~~ &\text{otherwise},
    \end{cases}
\end{aligned}
\end{equation}
which satisfies \cref{eq:acbf} if $h_a$ is a valid aCBF.
This modification expands the set of allowable states but the resulting controller is not necessarily Lipschitz and can exhibit high-frequency oscillations in closed-loop, as shown in \cref{ex:lipschitz}. 

\begin{example}
\label{ex:lipschitz}
Consider the uncertain system $\dot{x} = - \theta + u$ with $\theta>0$ and aCBF $h_a(x) = x - \ubar{x}$.
Let the controller $\kappa$ be the solution to $\kappa = {\argmin}~\frac{1}{2} u^2$ subject to $\dot{\bar{h}}_a(x) \geq 0$.
Then
\begin{equation}
\label{eq:u_ex}
\begin{aligned}
    \kappa = \begin{cases}
        0  ~ &\text{if} ~ h_a(x)\geq\sigma \\
        \mathrm{max}(0,\hat{\theta}) ~~ &\text{otherwise},
    \end{cases}
\end{aligned}
\end{equation}
where $\hat{\theta}$ is the estimate of $\theta$ and is modified based on \cite{taylor2019adaptive}
\begin{equation}
\label{eq:ha_ex}
\begin{aligned}
    \dot{\hat{\theta}} = \begin{cases}
        0  ~ &\text{if} ~ h_a(x)\geq\sigma \\
        - \Gamma \left[h_a(x)-\sigma\right] \frac{\partial h_a}{\partial x} ~~ &\text{otherwise},
    \end{cases}
\end{aligned}
\end{equation}
with $\Gamma > 0$. 
For $\hat{\theta}(0) = \hat{\theta}_0 \leq 0$ then $\kappa = 0$ so the closed-loop response is $x_{cl}(t) = - \theta (t-t_0) + x_0$ where $x_{cl}(t_0) = x_0 > \ubar{x}$ and $t \geq t_0$.
From the adaptation law \cref{eq:ha_ex}, it is easy to see that $\dot{\hat{\theta}} \geq 0 $ which will necessarily lead to $\hat{\theta} > 0$ since $h_a \leq \sigma$ until $\kappa > \theta$.
For $\hat{\theta} > 0$, the closed-loop response becomes
\begin{equation}
\label{eq:cl_ex}
    \begin{aligned}
        x_{cl}(t) = \begin{cases}
            -\theta (t-t_0) + x_0  ~ &\text{if} ~ h_a(x)\geq\sigma \\
            \hspace{0.28cm}\tilde{\theta} (t-t_0') + x_0' ~~ &\text{otherwise},
        \end{cases}
    \end{aligned}
\end{equation}
where $\tilde{\theta}:=\hat{\theta} - \theta$ and $x_{cl}(t_0') = x_0'$.
Since $\tilde{\theta} > 0$, \cref{eq:cl_ex} will continuously switch between its two solutions.
The control policy $\kappa$ must then also switch between $0$ and $\hat{\theta}$ based on \cref{eq:u_ex}.
Furthermore, $\kappa$ is not locally Lipschitz continuous and will exhibit high-frequency oscillations of magnitude $\hat{\theta}$ in closed-loop.
\end{example}

The intuition behind \cref{ex:lipschitz} is that chatter arises due to the barrier condition switching between being trivially satisfied, i.e., $\dot{\bar{h}}_a = 0 \geq 0$ for all $u$, to satisfied only for a particular $u$, i.e., a $u$ so that $\dot{\bar{h}}_a \geq 0$.
The approach developed in this work addresses the conservatism of aCBFs and results in a locally Lipschitz continuous controller.

\subsection{Robust aCBFs}
\label{sub:racbf}
This section will show that a \emph{tightened} set can be made forward invariant if the unknown model parameters are bounded and the parameter adaptation rate is an \emph{admissible} (to be defined) symmetric positive definite matrix.

\begin{assumption}
The unknown parameters $\theta$ belong to a known closed convex set $\Theta$.
The parameter estimation error $\tilde{\theta}:= \hat{\theta} - \theta$ then also belongs to a known closed convex set $\tilde{\Theta}$ and the maximum possible parameter error is $\tilde{\vartheta}$.
\end{assumption}

Let $\mathcal{C}^r_\theta$ be a family of superlevel sets parameterized by $\theta$ for a continuously differentiable function ${h}_r: \mathbb{R}^n \times \mathbb{R}^p \rightarrow \mathbb{R}$  
\begin{equation*}
    \begin{aligned}
    \mathcal{C}^r_\theta & = \left\{ x \in \mathbb{R}^n : {h}_r(x,\theta) \geq \frac{1}{2}\tilde{\vartheta}^\top \Gamma^{-1} \tilde{\vartheta}  \right\} \\
    \partial \mathcal{C}^r_\theta & = \left\{ x \in \mathbb{R}^n : {h}_r(x,\theta) = \frac{1}{2}\tilde{\vartheta}^\top \Gamma^{-1} \tilde{\vartheta} \right\} \\
    \text{Int}\left(\mathcal{C}^r_\theta\right) & = \left\{ x \in \mathbb{R}^n : {h}_r(x,\theta) > \frac{1}{2}\tilde{\vartheta}^\top \Gamma^{-1} \tilde{\vartheta} \right\},
    \end{aligned}
\end{equation*}
where $\Gamma \in \mathscr{S}^p_+ \subset \mathcal{S}^p_+$ is an admissible symmetric positive definite matrix that will dictate the parameter adaptation rate.
The set $\mathcal{C}^r_\theta$ can be viewed as a \emph{tightened} set with respect to $\mathcal{C}_\theta$, i.e., $\mathcal{C}^r_\theta \subset \mathcal{C}_\theta$, shown in \cref{fig:racbf}.
One can select the desired subset $\mathcal{C}^r_\theta$ to be made forward invariant \textit{a priori}  by choosing $h_r(x_r,\theta_r) >0 $ for appropriate $x_r,~\theta_r$ so ${h}_r(x_r,\theta_r) = \frac{1}{2}\tilde{\vartheta}^\top \Gamma^{-1} \tilde{\vartheta}$.
To reduce conservatism, one can either 1) have fast parameter adaptation or 2) reduce model error.
The first scenario can lead to well-known undesirable effects in practice so the second scenario is the most viable, and will be explored more in \cref{sec:smid}.
\cref{eq:unc_dyn} is again assumed to be locally Lipschitz so \cref{def:fi,def:safety} hold.

\begin{definition}
\label{definition:racbf}
Let $\mathcal{C}^r_\theta$ be a family of superlevel sets parameterized by $\theta$ for a continuously differentiable function $h_r:\mathbb{R}^n\times \mathbb{R}^p\rightarrow\mathbb{R}$, then $h_r$ is a \emph{robust adaptive control barrier function}  if there exists an extended class $\mathcal{K}_\infty$ function $\alpha$ such that for all $\theta\in\Theta$ 
\begin{equation}
\label{eq:racbf}
\begin{aligned}
    &\underset{u \in \mathcal{U}}{\text{sup}}~ \left[ \frac{\partial h_r}{\partial x}(x,\theta)\left[ f(x) - \Delta(x)^\top \Lambda(x,\theta) + B(x) u\right]\right] \\ 
    & \hspace{3cm} \geq -\alpha\left(h_r(x,\theta) - \frac{1}{2}\tilde{\vartheta}^\top \Gamma^{-1} \tilde{\vartheta}\right),
\end{aligned}
\end{equation}
where $\Lambda(x,\theta) := \theta - \Gamma \left(\frac{\partial h_r}{\partial \theta}(x,\theta)\right)^\top$, $\tilde{\vartheta}$ is the maximum possible parameter error, and $\Gamma \in \mathscr{S}^p_+ \subset \mathcal{S}^p_+$ is an admissible symmetric positive definite matrix.
\end{definition}

The invariance condition \cref{eq:racbf} is reminiscent of that in \cref{eq:cbf} and is less conservative than that in \cref{eq:acbf} because the system is allowed to approach the boundary of $\mathcal{C}^r_\theta$.
\cref{thm:racbf} shows the existence of a RaCBF, coupled with an adaptation law, renders the set $\mathcal{C}^r_\theta$ forward invariant and hence safe.

\begin{theorem}
\label{thm:racbf}
Let $\mathcal{C}^r_{\hat{\theta}}\subset\mathbb{R}^n$ be a superlevel set of a continuously differentiable function $h_r:\mathbb{R}^n \times \mathbb{R}^p\rightarrow\mathbb{R}$, if $h_r$ is a RaCBF on $\mathcal{C}^r_{\hat{\theta}}$, then any locally Lipschitz continuous controller satisfying \cref{eq:racbf} renders the set $\mathcal{C}^r_{\hat{\theta}}$ safe for the uncertainty system with adaptation law and adaptation gain
\begin{equation*}
\dot{\hat{\theta}} = \Gamma \Delta(x) \left(\frac{\partial h_r}{\partial x}(x,\hat{\theta}) \right)^\top, ~~ \lambda_{\min}(\Gamma) \geq \frac{\|\tilde{\vartheta}\|^2}{2h_r(x_r,{\theta}_r)},
\end{equation*}
where $\tilde{\vartheta}$ is the maximum possible parameter error, $h_r(x_r,\theta_r)>0$ can be chosen freely based on the desired conservatism, and $\Gamma \in \mathscr{S}^p_+ \subset \mathcal{S}^p_+$ is an admissible symmetric positive definite matrix.
Furthermore, the original set $\mathcal{C}_{\hat{\theta}}$ is also safe for the uncertain system.
\end{theorem}


%

\begin{remark}
\label{remark:proj}
The projection operator \cite{slotine1991applied} can be used to enforce parameter bounds by modifying the above adaptation law as opposed to capturing them explicitly with $h_r$. 
This can simplify the design of $h_r$ without forfeiting safety.
The proof is omitted but one can show that a positive semi-definite term appears in the same composite candidate CBF used in \cref{thm:racbf} when adaptation is temporarily stopped.
\end{remark}

\begin{figure}[t!]
\vskip 0.1in
    \begin{subfigure}{.48\columnwidth}
         \centering
         \includegraphics[trim=100 0 100 0, clip,width=1\linewidth]{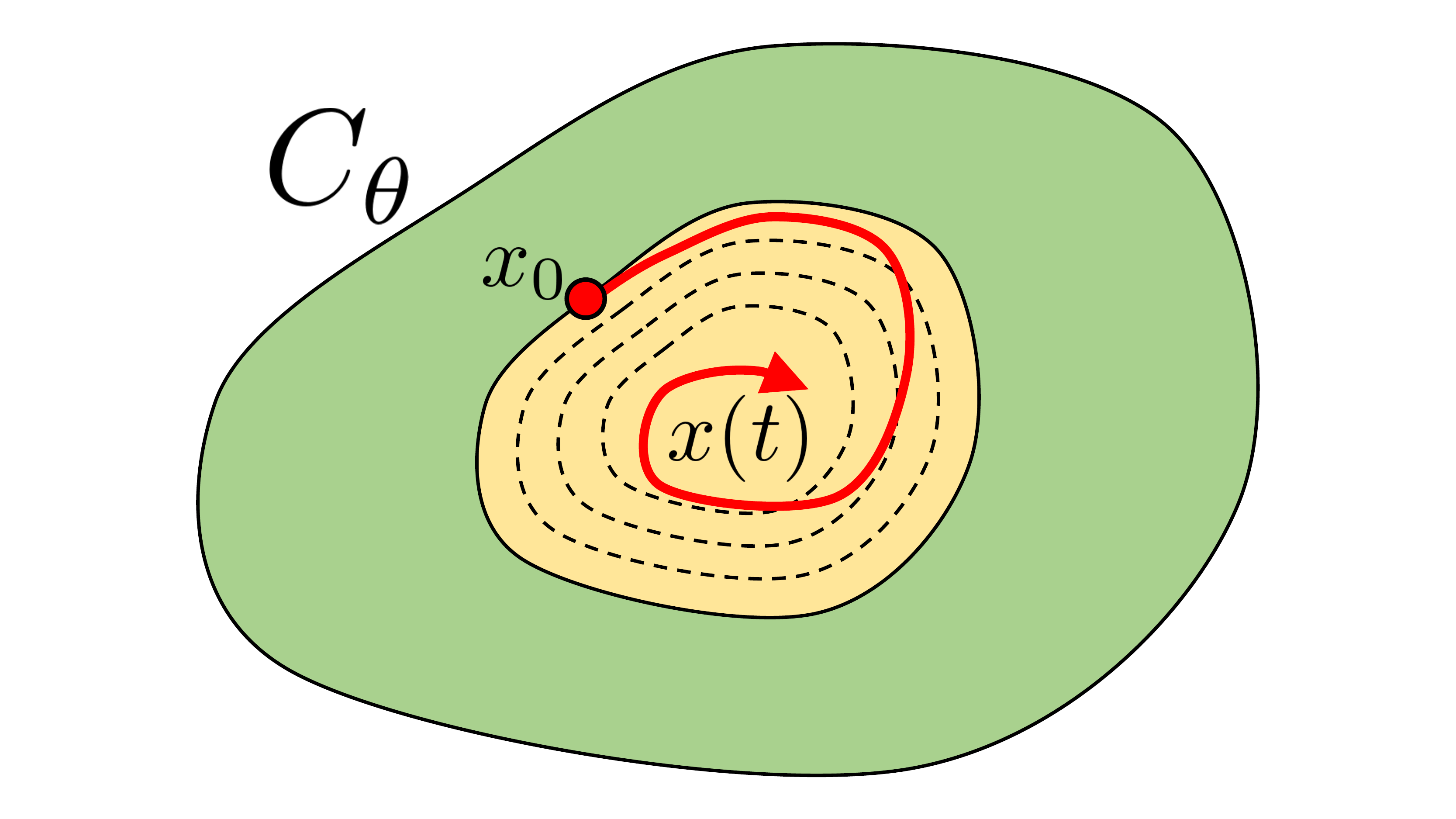}
         \caption{{Safe set with adaptive control barrier functions (aCBFs).}}
         \label{fig:acbf}
     \end{subfigure}
    \hspace{0.3em}
     \begin{subfigure}{.48\columnwidth}
         \centering
         \includegraphics[trim=100 0 100 0, clip,width=1\linewidth]{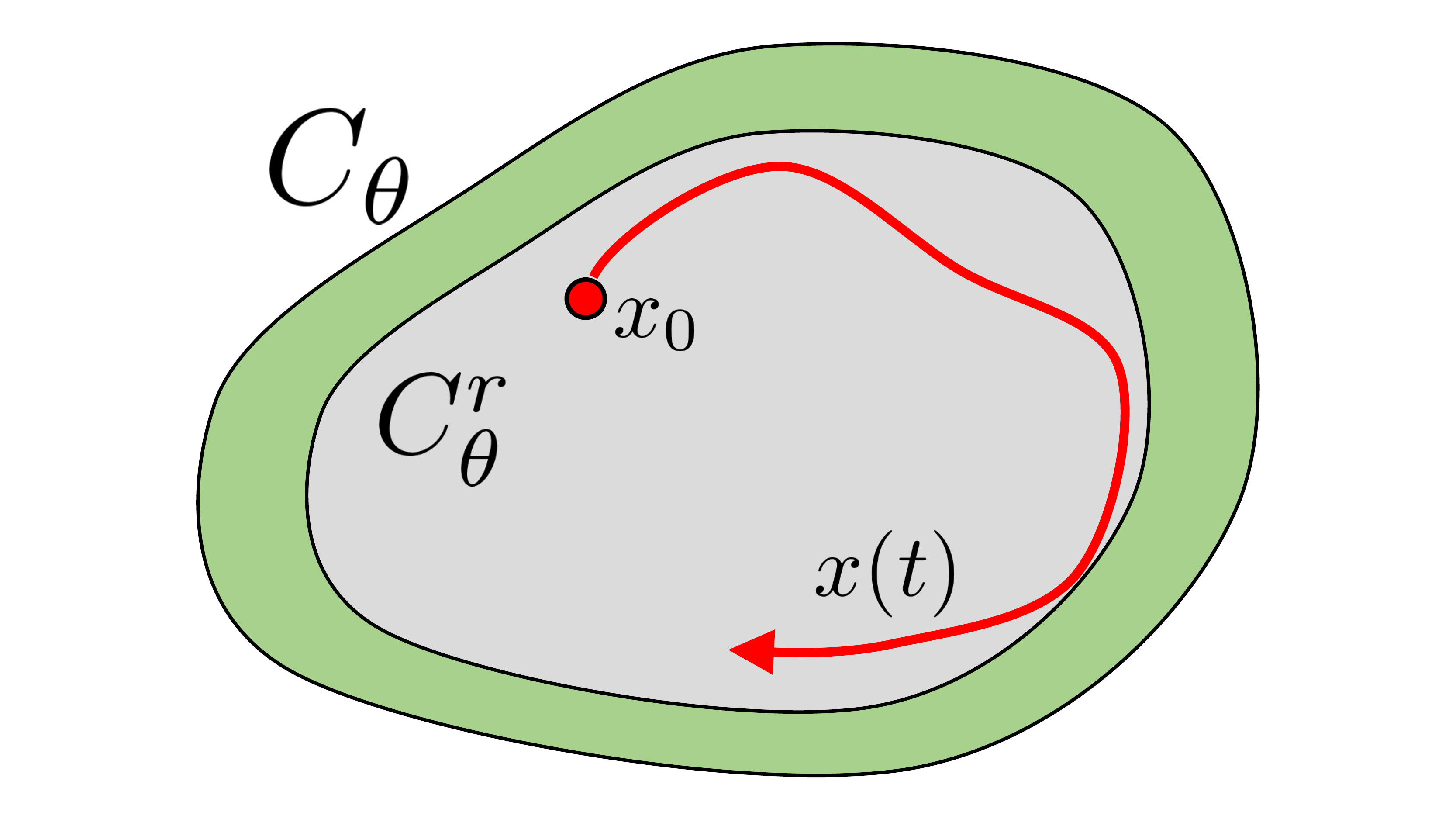}
         \caption{{Safe set with robust adaptive control barrier functions (RaCBFs).}}
         \label{fig:racbf}
     \end{subfigure}
    \caption{Visual comparison of safe sets with adaptive and robust adaptive control barrier functions. (a): System is restricted to level sets (black dashed lines) of aCBF $h_a$. (b): System allowed to operate in larger set with RaCBF $h_r$ reducing conservatism.}
    \vskip -0.25in
\end{figure}

Several remarks can be made about \cref{thm:racbf}.
First, safety is \emph{guaranteed} for all possible parameter realizations through adaptation with minimal conservatism.
Hence, RaCBFs expand and improve the \emph{adaptive safety} paradigm.
Second, the minimum eigenvalue condition for the adaptation rate depends on the desired conservatism, i.e., the degree of tightening by choice of $h_r(x_r,\theta_r)$.
For low conservatism, i.e., a small $h_r(x_r,\theta_r)$ value, the adaptation rate must be large so the parameter estimates can change quickly to ensure forward invariance of $\mathcal{C}^r_\theta$.
There is thus a fundamental trade-off between conservatism and parameter adaptation rate that must be weighed carefully given the well-known undesirable effects of high-gain adaptation.
Third, the RaCBF condition in \cref{eq:racbf} can be used as a safety filter for an existing tracking controller or as a constraint within an optimization.
\cref{sec:result} will show the latter but with a contraction-based controller.
Lastly, if the adaptation gain must be small (or the maximum parameter error is large) then RaCBFs can be conservative albeit not to the same extent as aCBFs.
Better performance can be obtained if the model parameters can be robustly and accurately estimated.
Instead of obtaining a point-estimate of the parameters, this work will instead identify the \emph{set} of possible parameter values.

\subsection{RaCBFs with Set Membership Identification}
\label{sec:smid}
Set membership identification (SMID) is a model estimation technique that constructs an unfalsified set of model parameters.
SMID was originally developed to identify transfer functions for uncertain linear systems \cite{pararrieter1992set}, but has been more recently applied to linear \cite{tanaskovic2014adaptive,lorenzen2017adaptive} and nonlinear adaptive MPC \cite{lopez2019adaptive}.
Assume that the true parameters $\theta^*$ belong to an initial set of possible parameters $\Theta^0$, i.e., $\theta^* \in \Theta^0$.
Given $k$ state, input, and rate measurements (denoted as $x_{1:k}$ and so forth), a set $\Xi$ can be constructed such that
\begin{equation*}
\Xi = \left\{ \varrho :  |\dot{x}_{1:k} - f_{1:k} + \Delta_{1:k}^\top \varrho - B_{1:k}u_{1:k}| \leq D  \right\},
\end{equation*}
where $D$ can be treated as a tuning parameter that dictates the conservativeness of SMID. 
It can also represent a disturbance or noise bound \cite{tanaskovic2014adaptive,lopez2019adaptive}.
The set of possible parameter values can then be updated via $\Theta^{j+1} = \Theta^j \cap \Xi$ for all $j\geq 0$.
In practice, $\Xi$ can be found by solving a linear program and set intersection can be efficiently done through a combination of min and max operations.
Restricting $\hat{\theta} \in \Theta^j$ then $\tilde{\theta} \in \Tilde{\Theta}^j$ where $\Tilde{\Theta}^j$ is the set of possible parameter errors.
The following lemma shows the advantage of performing set identification over point-estimation techniques.
\begin{lemma}
\label{lemma:smid}
Model uncertainty monotonically decreases with set membership identification, i.e., $\tilde{\Theta}^{j+1} \subseteq \tilde{\Theta}^j$ for all $j\geq 0$.
\end{lemma}
\begin{proof}
Since $\theta^* \in \Theta^{j+1}$ then $\theta^* \in \Theta^{j} \cap \Xi$ which is true if $\theta^* \in \Theta^j$ so $\Theta^{j+1}\subseteq\Theta^j$ and $\tilde{\Theta}^{j+1} \subseteq \tilde{\Theta}^j$.
\end{proof}

The motivation to combine SMID with RaCBFs is to enlarge the tightened set $\mathcal{C}_\theta^r$.
To do so, one must ensure $\mathcal{C}^r_\theta$ remains forward invariant as the set of model parameters is updated.
In general this is non-trivial to prove since the maximum possible parameter error is now time varying.
However, \cref{thm:racbf_smid} shows that safety is maintained if the model uncertainty monotonically decreases.

\begin{theorem}
\label{thm:racbf_smid}
Let $\mathcal{C}^r_\theta$ be a superlevel set of a continuously differentiable function ${h}_r : \mathbb{R}^n\times \mathbb{R}^p \rightarrow \mathbb{R}$. If the system is safe on $\mathcal{C}^r_\theta$ then it remains safe if the maximum allowable model parameter error $\tilde{\vartheta}$ monotonically decreases. 
Moreover, the tightened set $\mathcal{C}^r_\theta$ converges to $\mathcal{C}_\theta$ monotonically.
\end{theorem}


Combining RaCBFs and SMID provides a mechanism to 1) modify parameters via adaptation to achieve safety and 2) update the model to reduce uncertainty and conservatism.
Safety is \emph{guaranteed} even as the model parameters are modified online, and the system's performance will only improve as more data is collected.
This \emph{adaptive data-driven safety} paradigm can be merged with a stabilizing adaptive controller for safe reference tracking.
To maximize the generality of the proposed unification, the adaptive controller must be applicable to a broad class of nonlinear systems.

\section{Adaptive Control with Contraction Metrics}
\label{sec:accm}
Several adaptive control techniques have been proposed for nonlinear systems, including methods based on feedback linearization, variable structure, and backstepping (see \cite{slotine1991applied,miroslav1995nonlinear}).
These methods are limited to certain classes of systems because they rely on \emph{explicitly} constructing a control Lyapunov function (CLF) to prove stability.
This work will instead utilize a \emph{differential} approach based on contraction analysis that can be applied to a broad class of systems.

\begin{figure}[t!]
\vskip 0.05in
    \begin{subfigure}{.47\columnwidth}
         \centering
         \includegraphics[width=1\linewidth]{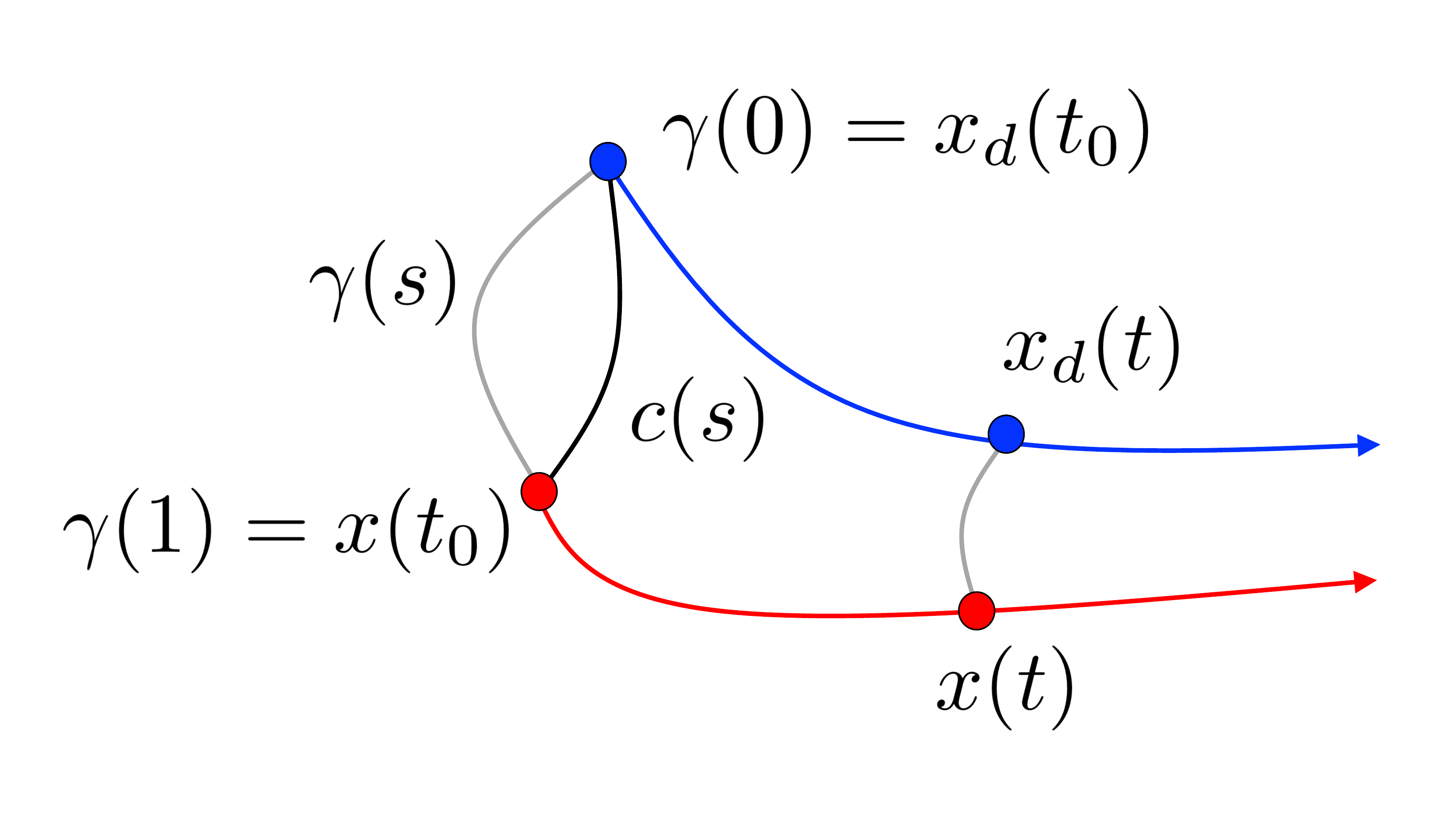}
         \caption{{Geodesic and arbitrary curve connecting current and desired state.}}
         \label{fig:geodesig}
     \end{subfigure}
    \hspace{0.6em}
     \begin{subfigure}{.47\columnwidth}
         \centering
         \includegraphics[width=1\linewidth]{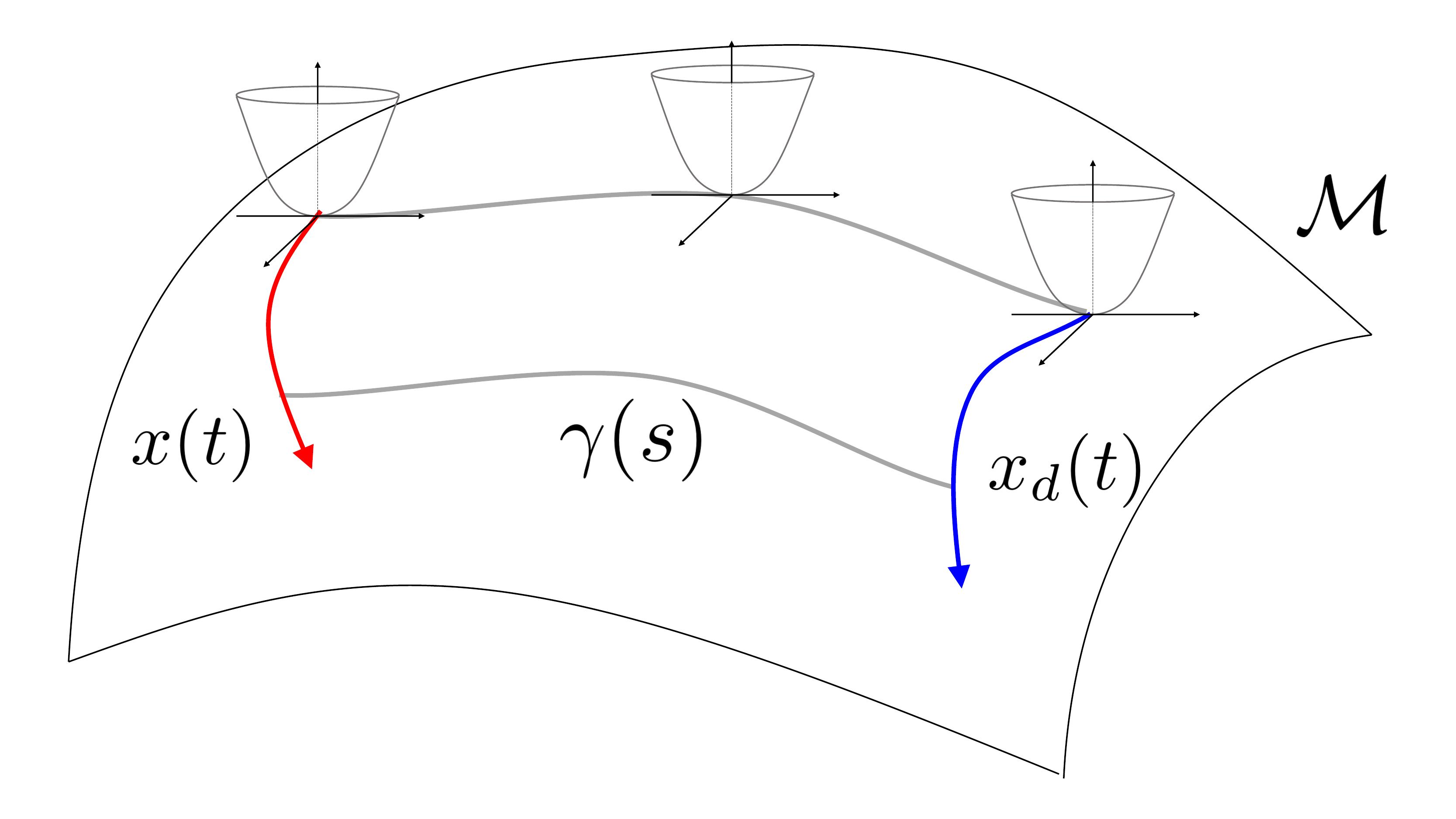}
         \caption{{Differential CLFs along geodesic connecting current and desired state.}}
         \label{fig:manifold}
     \end{subfigure}
    \caption{Geodesic and differential CLF visualization. (a): Geodesic (grey) connecting current $x$ (red) and desired $x_d$ (blue) state. (b): Differential CLFs are integrated along geodesic to achieve exponential convergence.}
    \label{fig:dclf}
    \vskip -0.2in
\end{figure}

\subsection{Contraction Metrics}
The nominal differential dynamics of \cref{eq:unc_dyn} are $\dot{\delta}_x  = A(x,u)\delta_x + B(x)\delta_u$ where $A(x,u) = \frac{\partial f}{\partial x} + \sum_{i=1}^m \frac{\partial b_i}{\partial x}u_i$.
Contraction analysis searches for a \emph{control contraction metric} (CCM) $M(x)$ such that a \emph{differential} CLF $\delta V = \delta_x^\top M(x) \delta_x$ satisfies $\delta \dot{V}\leq - 2 \lambda \delta V$ for all $x$.
A global CLF can be obtained by integrating along a geodesic $\gamma(s)$, illustrated in \cref{fig:dclf}, with $\gamma(0) = x_d$ and $\gamma(1) = x$ where $x_d$ and $x$ are the desired and current state.
The Riemannian energy and tracking error both converge to zero exponentially, i.e., $\dot{E} \leq - 2 \lambda E$.
Let $W(x) = M(x)^{-1}$, $M(x)$ is CCM if \cite{manchester2017control}
\begin{gather}
     B^\top_{\perp}\left( W A ^\top + A W - \dot{W} + 2 \lambda W \right)B_{\perp} \preceq 0 \tag{C1} \label{eq:ccm_c1} \\
    \partial_{b_i}W - W \frac{\partial b_i}{\partial x}^\top  -  \frac{\partial b_i}{\partial x} W = 0 ~~ i=1,\dots,m \tag{C2} \label{eq:ccm_c2}
\end{gather}
where $B_\perp(x)$ is the annihilator matrix of $B(x)$, i.e., $B^\top_\perp B = 0$.
\cref{eq:ccm_c1} ensures the dynamics orthogonal to $u$ are contracting and is a stabilizability condition.
\cref{eq:ccm_c2} requires the column vectors of $B(x)$ form a Killing vector for the dual metric $W(x)$ leading to simpler controllers \cite{manchester2017control}. 
\cref{eq:ccm_c1} and \cref{eq:ccm_c2} will be referred to as the \emph{strong CCM conditions}.

\subsection{Adaptive Control \& Contraction}
A novel adaptive control method was developed in \cite{lopez2019contraction} for closed-loop contracting systems with \emph{extended matched uncertainties}, i.e., $\Delta(x)^\top \theta \in \text{span}\{B,ad_fB\}$ were $ad_fB$ is the Lie bracket of the vector fields $f(x)$ and $B(x)$.
To stabilize such systems, the \emph{parameter-dependent} metric $M(x,\hat{\theta})$ was introduced and must satisfy the strong CCM conditions \emph{for all} possible $\hat{\theta}$.
This led to the following result.

\begin{theorem}[\cite{lopez2019contraction}]
\label{thm:exMatched}
If a parameter-dependent metric can be computed for \cref{eq:unc_dyn} with extended matched uncertainties, then the closed-loop systems is asymptotically stable with
\begin{equation}
    \dot{\hat{\theta}} = - \Gamma \Delta(x) M(x,\hat{\theta})\gamma_s(1)
    \label{eq:accm}
\end{equation}
where $\gamma_s(s):=\frac{\partial \gamma}{\partial s}$ is the geodesic speed and $\Gamma \in \mathcal{S}^p_+$ is a symmetric positive definite matrix.
\end{theorem}

\begin{remark}
For \emph{matched} uncertainties, the metric is \emph{independent} of the unknown parameters $\hat{\theta}$ \cite[Lemma 1]{lopez2019contraction} simplifying its computation.
Otherwise, sum-of-square or robust optimization must be utilized to compute $M(x,\hat{\theta})$.
\end{remark}

\begin{remark}
Several modifications can be made to \cref{eq:accm} that improve transients or robustness including the projection operator discussed in \cref{sub:racbf} (see \cite{lopez2019contraction}).
\end{remark}

\subsection{Offline Design \& Online Computation}
\label{sub:computation}
A contraction metric is computed offline via sum-of-square programming for polynomial systems \cite{manchester2017control} or by imposing \cref{eq:ccm_c1,eq:ccm_c2} at sampled points in the state space; a process called gridding.
Geodesics are computed online at each time step by solving a nonlinear program (NLP) with the current state.
Geodesics are often guaranteed to exist by Hopf-Rinow and are less expensive to compute than solving nonlinear MPC.
Given a geodesic $\gamma(s)$, the Riemannian energy can be interpreted as a CLF so a pointwise min-norm controller similar to that in \cite{primbs2000receding} can found by solving the QP
\begin{equation*}
    \begin{aligned}
        & u^* = ~\underset{u\in \mathcal{U}}\argmin ~\frac{1}{2} u^\top u \\
        & \text{s.t.} ~\gamma_s(1)^\top M(x,\hat{\theta})\dot{\hat{x}} - \gamma_s(0)^\top M(x_d,\hat{\theta})\dot{x}_d  \leq - \lambda E(\gamma(s),\hat{\theta}) 
    \end{aligned}
\end{equation*}
where $\hat{\theta}$ is the current parameter estimate, $\gamma_s$ is the geodesic speed, $\dot{\hat{x}}$ is \cref{eq:unc_dyn} but with $\hat{\theta}$, $\dot{x}_d $ is the desired dynamics for desired state $x_d$.
The stability constraint imposes $\dot{E} \leq - 2 \lambda E$ where $\dot{E}$ is the first variation of the Riemannian energy, which has a known form \cite{carmo1992riemannian}, and results in exponential convergence of the tracking error.
It is more general than the traditional Lyapunov stability as a distance-like function does not need to computed explicitly.
Safety can be directly embedded in the above QP, resulting in a single optimization for a safe stabilizing controller.

\begin{figure*}[t]
\vskip 0.1in
    \begin{subfigure}{.66\columnwidth}
         \centering
         \includegraphics[trim=150 30 150 30, clip,width=1\textwidth]{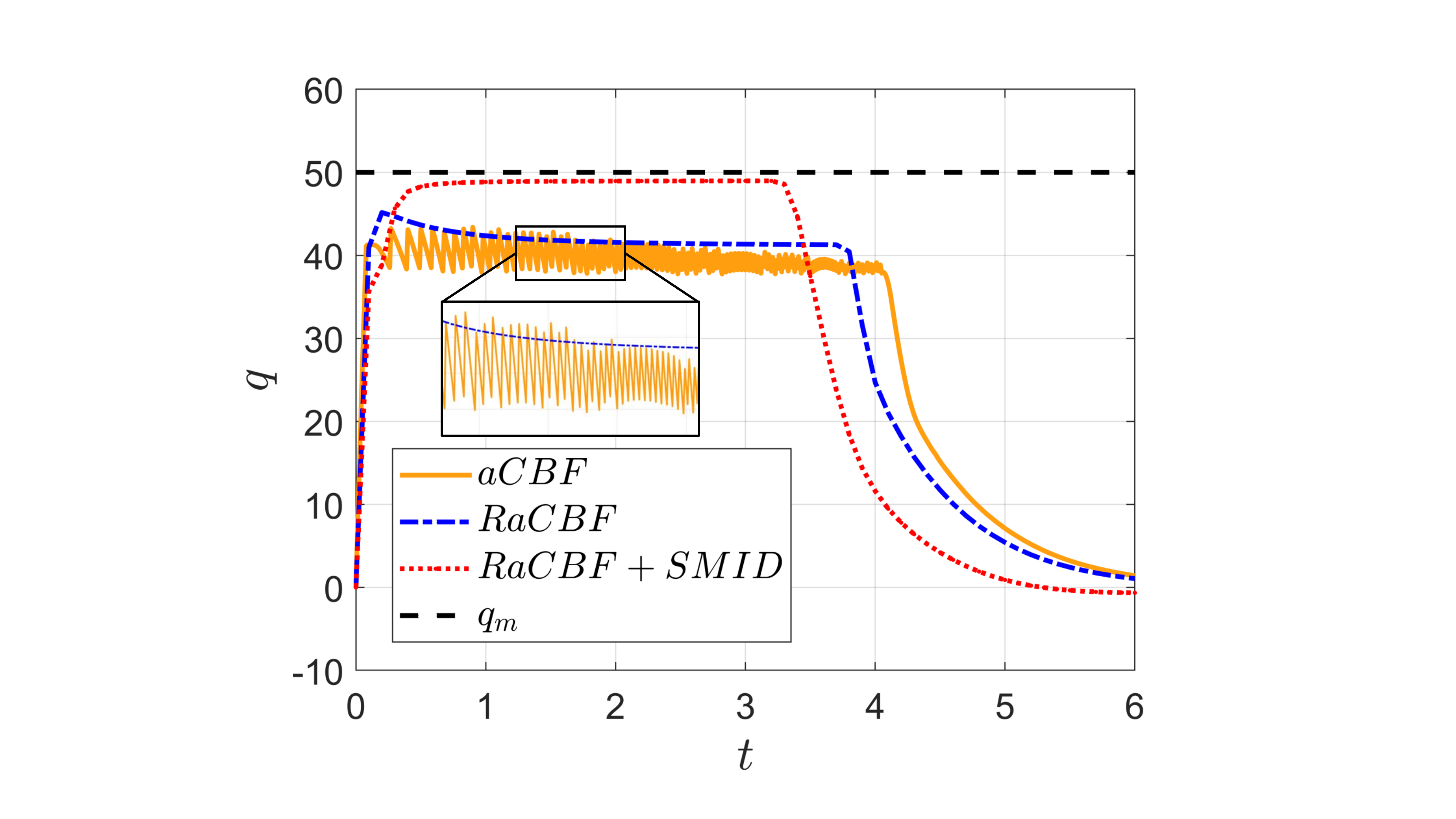}
         \caption{Pitch rate $q$.}
         \label{fig:q_i}
     \end{subfigure}
     \hfill{}
     \begin{subfigure}{.66\columnwidth}
         \centering
         \includegraphics[trim=150 30 150 30,clip, width=1\textwidth]{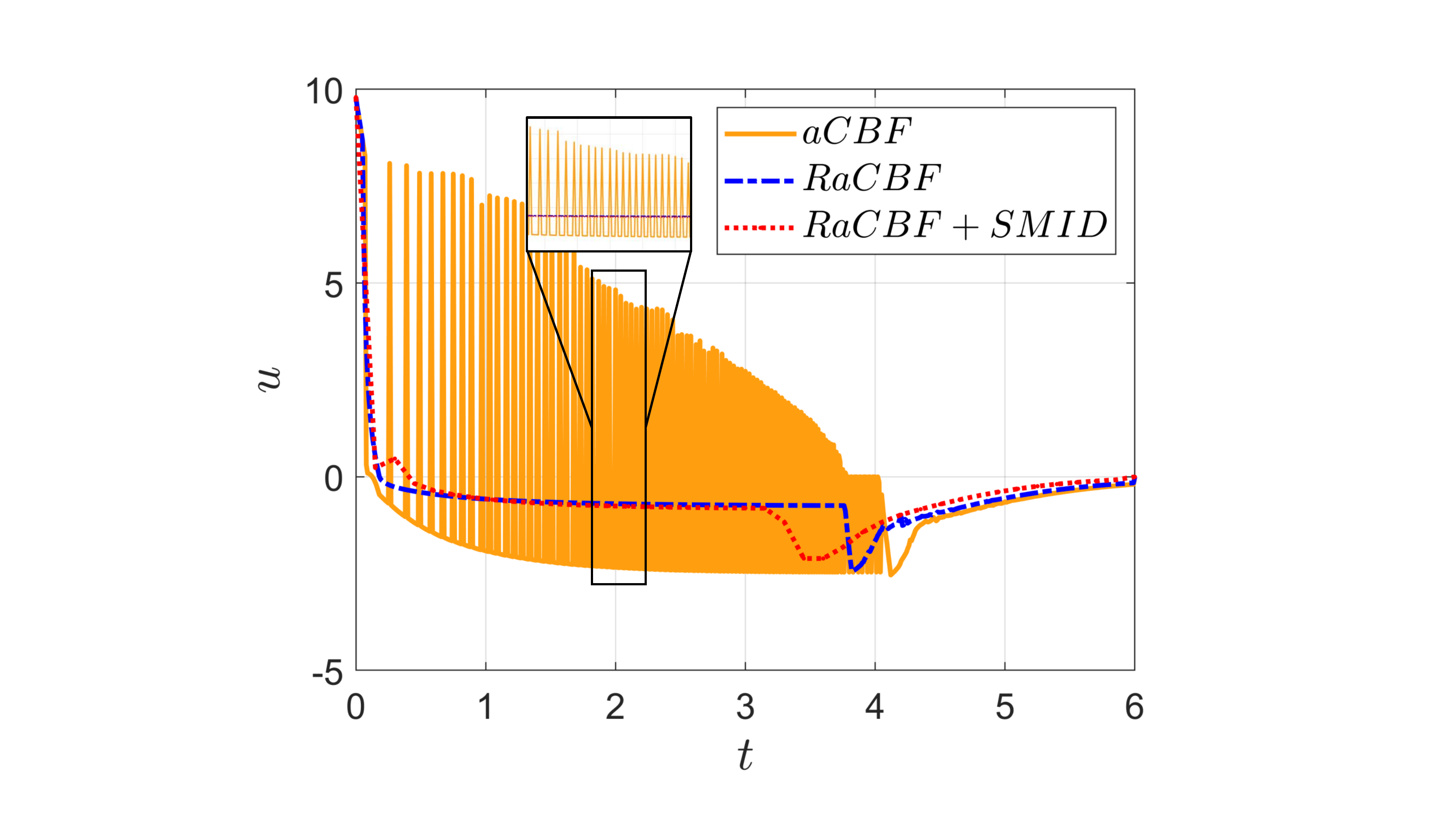}
         \caption{Control input $u$.}
         \label{fig:u_i}
     \end{subfigure}
     \hfill{}
     \begin{subfigure}{.66\columnwidth}
         \centering
         \includegraphics[trim=150 30 150 30,clip, width=1\textwidth]{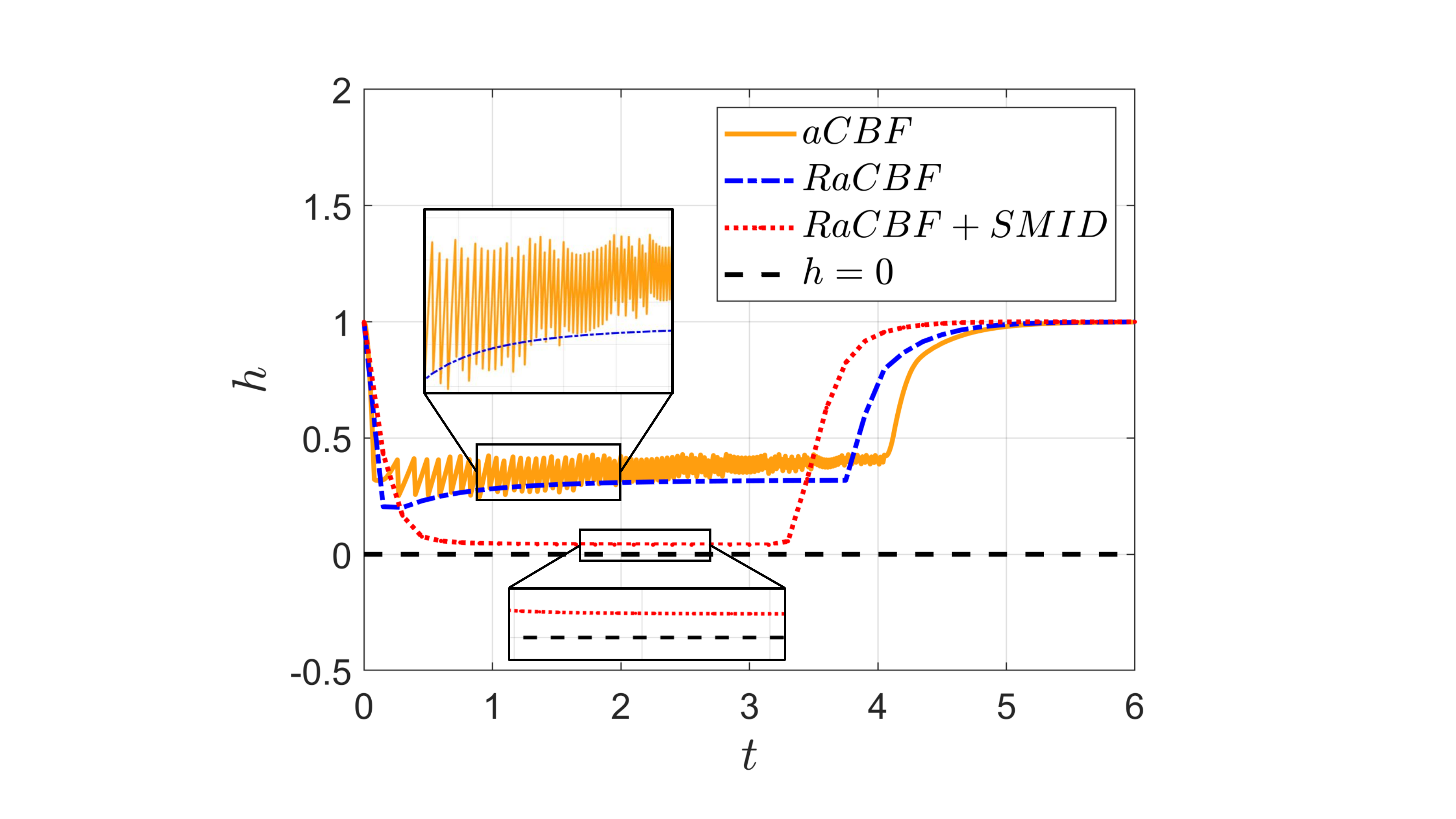}
         \caption{Barrier function $h$.}
         \label{fig:h_i}
     \end{subfigure}
     \caption{Comparison of modified aCBFs \cref{eq:acbf_m}, RaCBFs, and RaCBFs \& SMID for desired terminal state $x_d = [180^\circ~0~0]^\top$ (Immelmann turn) and maximum pitch rate $q_m$. (a): Pitch rate tracking where RaCBFs and the modified aCBFs exhibit similar conservatism due to model error.
     RaCBF \& SMID allows the aircraft to utilize 97.9\% of the maximum allowable pitch rate. `
     (b): Control chatter is observed with the modified aCBFs while RaCBFs generate continuous control inputs. (c): Safety is maintained but RaCBFs and the modified aCBFs are conservative due to potential model error. RaCBFs \& SMID permit states closer to the boundary of the safe set without losing safety guarantees.
     For tests $k_q^* = 0.2$, $\ell^*_\alpha = -1$, $\Gamma_B = 20$, $\Gamma_C = 50$, $\alpha(r) = 10r$, and $D = 0.1$.}
     \label{fig:results_immelmann}
\vskip -0.2in
\end{figure*}

\section{Adaptive \& Data-Driven Safety}
\label{sec:result}
A safe and stabilizing controller can be computed by unifying RaCBFs, SMID, and adaptive control with contraction.
The individual components of the controller are summarized below with their respective computational complexity.
\\[4pt]
\noindent \textbf{1) Compute geodesic (NLP)}
\begin{equation*}
    \gamma(s) = \underset{c(s)\in\Upsilon(x,x_d)}{\argmin}~ E(c,\hat{\theta}_C) = \int_0^1 c_s^\top M(c,\hat{\theta}_C) c_s ds
\end{equation*}

\noindent \textbf{2) Compute controller (QP \& Quadrature)}
\begin{align*}
    & \kappa  = ~\underset{u \in \mathcal{U}}{\argmin} ~ \frac{1}{2} u^\top u + r \epsilon^2 \\
    &\text{s.t.} ~\gamma_s(1)^\top M(x,\hat{\theta}_C)\dot{\hat{x}} - \gamma_s(0)^\top M(x_d,\hat{\theta}_C)\dot{x}_d  \\
    & \hspace{4.6cm} \leq - \lambda E(\gamma(s),\hat{\theta}_C)+ \epsilon \nonumber  \\
    & \hspace{.4cm} \frac{\partial h_r}{\partial x}(x,\hat{\theta}_B)\left[ f(x) - \Delta(x)^\top \Lambda(x,\hat{\theta}_B) + B(x) u\right] \\
    & \hspace{3.1cm} \geq -\alpha\left({h}_r(x,\hat{\theta}_B) - \frac{1}{2}\tilde{\vartheta}^\top \Gamma^{-1} \tilde{\vartheta}\right) \nonumber
\end{align*}
\noindent \textbf{3) Update parameters (Quadrature)} \label{item:quad}
\begin{equation*}
\begin{aligned}
    & \dot{\hat{\theta}}_C = - \Gamma_C \Delta(x) M(x,\hat{\theta}_C) \gamma_s(1) \\
    & \dot{\hat{\theta}}_B =  \Gamma_B \Delta(x) \left(\frac{\partial h_r}{\partial x}(x,\hat{\theta}_B)\right)^\top
\end{aligned}
\end{equation*}
\vskip -.03in
\noindent \textbf{4) Update parameter error bounds (LP)} \label{item:lp}
\begin{align*}
        & \Xi = \left\{ \varrho : ~ |\dot{x}_{1:k} - f_{1:k} + \Delta_{1:k}^\top \varrho - B_{1:k}\kappa_{1:k}| \leq D  \right\} \\
         &\Theta^{j+1} = \Theta^{j} \cap \Xi, ~~ \tilde{\vartheta} = \underset{\varrho_i,\forall i}{\text{sup}} ~\Theta^{j+1} - \underset{\varrho_i,\forall i}{\vphantom{\text{sup}}\text{inf}}~ \Theta^{j+1} \nonumber
\end{align*}
\vskip -.02in

The NLP in Step~1) can be efficiently solved by parameterizing geodesics with a set of polynomial basis functions. 
We adopt the same strategy as in \cite{leung2017nonlinear} and utilize the the Chebychev Pseudospectral method and Clenshaw-Curtis quadrature to compute a geodesic at each time step.
Using the geodesic computed in Step~1), the QP in Step~2) is solved to generate a safe and stabilizing controller $\kappa$.
The QP is similar to that in \cite{ames2016control} but the stability constraint is replaced with the first variation of the Riemannian energy \cite{carmo1992riemannian}.
Under the premise that \cref{eq:unc_dyn} is locally Lipschitz, one can show that $\kappa$ is guaranteed to be locally Lipschitz from \cite[Theorem 3]{ames2016control} as both the geodesic speed $\gamma_s$ and metric $M(x)$ are also locally Lipschitz from their definitions.
Note that $\dot{x}_d$ is the desired dynamics and $\dot{\hat{x}}$ is \cref{eq:unc_dyn} but with $\hat{\theta}_C$.
Step~3) is simple quadrature and is not computationally expensive.
Note that the parameter adaptation $\dot{\hat{\theta}}_C$ for the controller should be temporarily stopped when the safety constraint is active to prevent undesirable transients.
Otherwise, the parameter estimates will windup as the tracking error may increase to ensure safety.

The LP in Step~4) has $2k$ constraints and is solved $2p$ times at every time step for the upper and lower bound of each parameter. 
Set intersection is done by taking the appropriate minimum or maximum of the newest and current bounds.
The complexity of Step~4) can be bounded by either removing redundant constraints or terminating when $\tilde{\vartheta} \leq \varepsilon$ where $\varepsilon$ is a predefined threshold.
Step~4) can also be done outside the control loop since stability and safety do not rely on real-time updates of the parameter bounds; although real-time bounds are desirable to quickly eliminate conservatism.
Consequently, non-causal filtering can be used to accurately estimate $\dot{x}$ if necessary.
Moreover, the right hand side of the inequality can be replaced by $D + \mathcal{E}$ where $\mathcal{E}$ is the maximum estimation error of the rate vector, i.e., $|\dot{\hat{x}}| \leq \mathcal{E}$.
The proposed method was tested in MATLAB R2018B with the built-in solvers without any code optimization on a 1.6GHz Intel i5 processor.
The NLP was initialized with a linear curve $c(s) = (x-x_d)s+x_d$ at each time step.



\section{Illustrative Example}
\label{sec:example}
Consider the simplified pitch dynamics of an aircraft \cite{mccormick1995aerodynamics}
\begin{equation*}
    \left[ \begin{array}{c} \dot{\theta} \\ \dot{\alpha} \\ \dot{q} \end{array} \right] = \left[ \begin{array}{c}  q \\ q - \bar{L}(\alpha)  \\ -k_q q + \bar{M}(\alpha) \end{array} \right] + \left[ \begin{array}{c} 0 \\ 0 \\ 1 \end{array} \right] u,
\end{equation*}
where $\theta$, $\alpha$, and $q$ are the pitch angle, angle of attack, and pitch rate.
$\bar{L}(\alpha)$ and $\bar{M}(\alpha)$ are the aerodynamic lift and moment.
The system is not feedback linearizable as the controllability matrix drops rank at $\bar{L}'(\alpha)=0$ and is not in strict-feedback form.
Utilizing flat plat theory \cite{mccormick1995aerodynamics}, the aerodynamics of a high-performance aircraft are approximately $\bar{L}(\alpha) = 0.8\mathrm{sin}(2\alpha)$ and $\bar{M}(\alpha) = -\ell_{\alpha} \bar{L}(\alpha)$.
The parameters $k_q$ and $\ell_\alpha$ are unknown but $k_q \in [0.1~0.8]$ and $\ell_\alpha \in [-3~1]$.
A metric quadratic in $\alpha$ was synthesized via gridding for $\alpha \in [-5^{\circ}~50{^\circ}]$ and $q \in [-10^{\circ/s}~50^{\circ/s}]$.
Note that $\bar{L}'(\alpha)=0$ is in the chosen grid range.
The function $h_r(q) = q_m - q$ where $q_m = 50^{\circ/s}$ can be easily shown to be a valid RaCBF that enforces $q \leq q_m$.

The desired terminal state $x_d = [180^\circ~0~0]^\top$ corresponds to the first portion of the aerobatic maneuver known as the Immelmann turn.
The vehicle executes a half loop before executing a half roll (not considered here) resulting in level flight in the opposite direction.
\cref{fig:q_i} shows modified aCBFs and RaCBFs exhibit similar behavior in terms of conservatism as they do not utilize the maximum allowable pitch rate.
However, modified aCBFs exhibit high-frequency oscillations due to the chatter in the control input, seen in \cref{fig:u_i}; a result of the formulation as the safety constraint continuously switches between active and inactive. 
High-frequency oscillations are also seen in the barrier function in \cref{fig:h_i} but are absent with RaCBFs.
\cref{fig:q_i} shows RaCBFs with SMID results in less conservatism as 97.9\% of the maximum allowable pitch rate is utilized.
Moreover, the set of allowable states is considerably larger as is evident by the small value of the barrier function in \cref{fig:h_i}.
The parameter error bounds, shown in \cref{fig:smid_i}, were reduced by 63.0\% for $k_q$ and 90.5\% for ${\ell}_\alpha$.
\cref{fig:comp_time_i} shows the computation time is within real-time constraints and can be easily reduced by utilizing faster solvers or a lighter programming language.

\begin{figure}[t!]
    \begin{subfigure}{.47\columnwidth}
         \centering
         \includegraphics[width=1\linewidth]{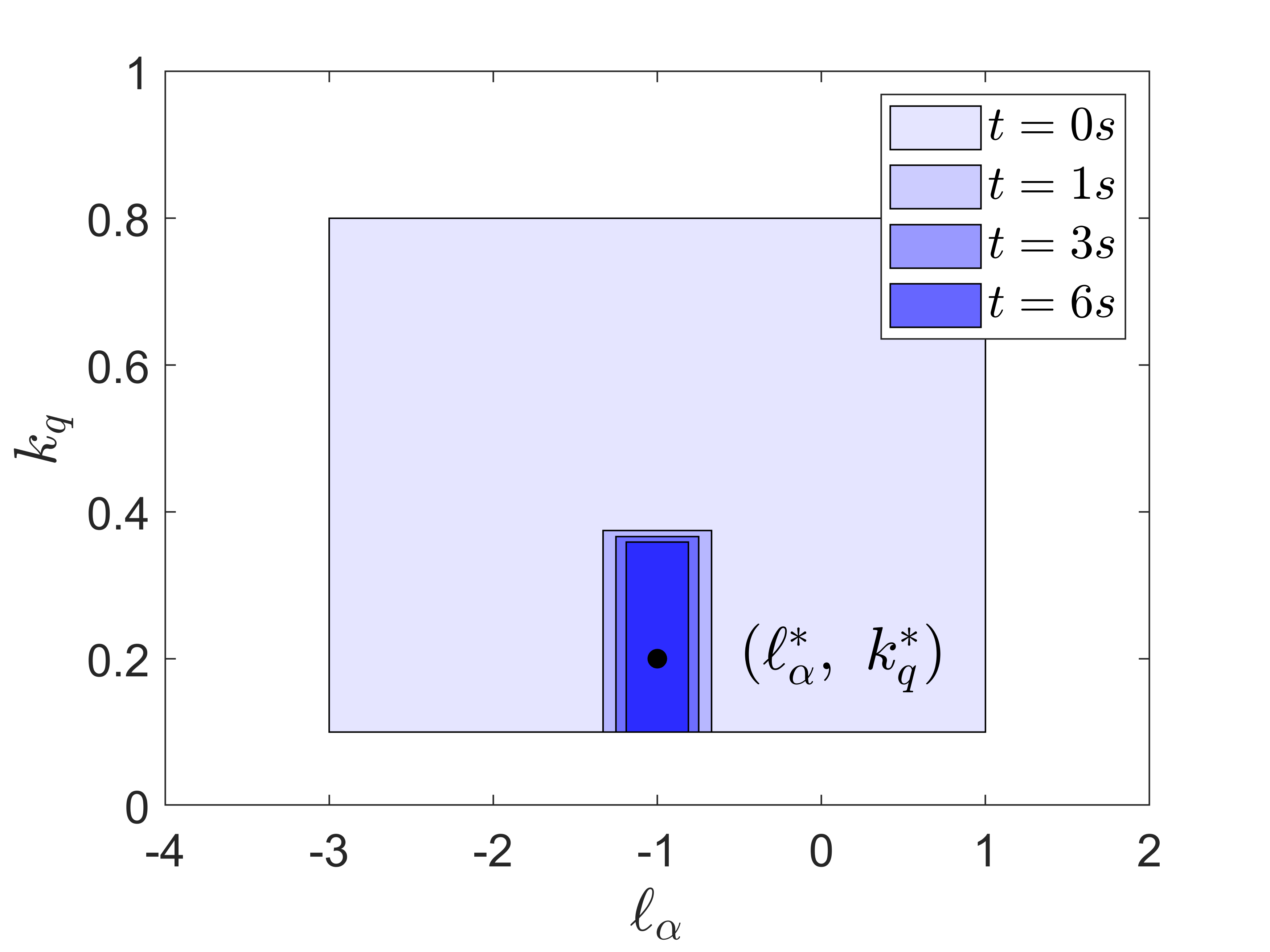}
         \caption{{Parameter bounds.}}
         \label{fig:smid_i}
     \end{subfigure}
    \hspace{0.6em}
     \begin{subfigure}{.47\columnwidth}
         \centering
         \includegraphics[width=1\linewidth]{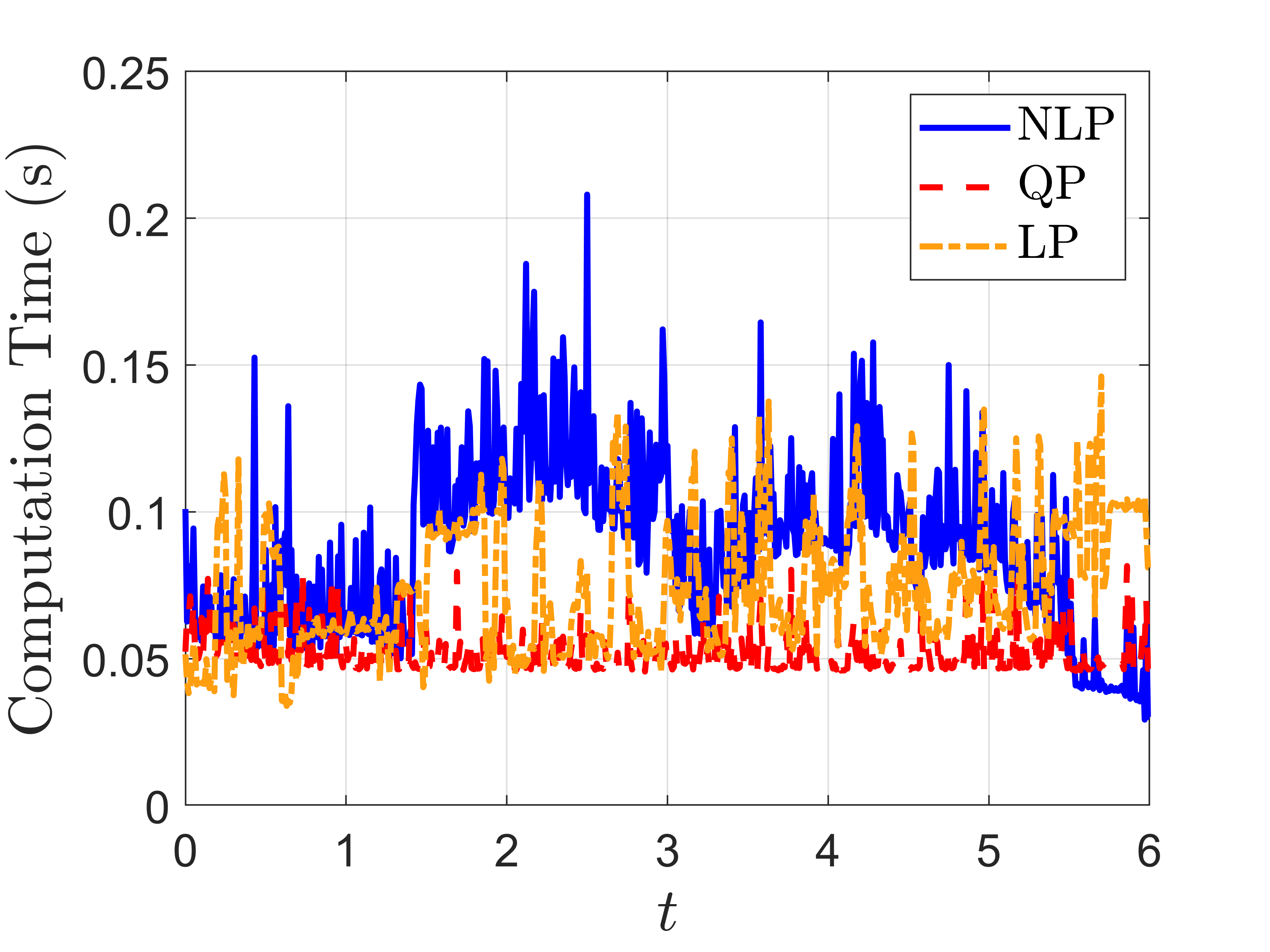}
         \caption{{Computation time.}}
         \label{fig:comp_time_i}
     \end{subfigure}
    \caption{Parameter bounds and computation time. (a): Parameter bounds monotonically approach the true parameter values $\ell_\alpha^*$ and $k_q^*$. (b): Computation time for the proposed controller is well within real-time constraints.}
    \label{fig:results_comp_i}
    \vskip -0.2in
\end{figure}

\begin{figure*}[t]
\vskip 0.1in
    \begin{subfigure}{.66\columnwidth}
         \centering
         \includegraphics[trim=150 30 150 30, clip,width=1\textwidth]{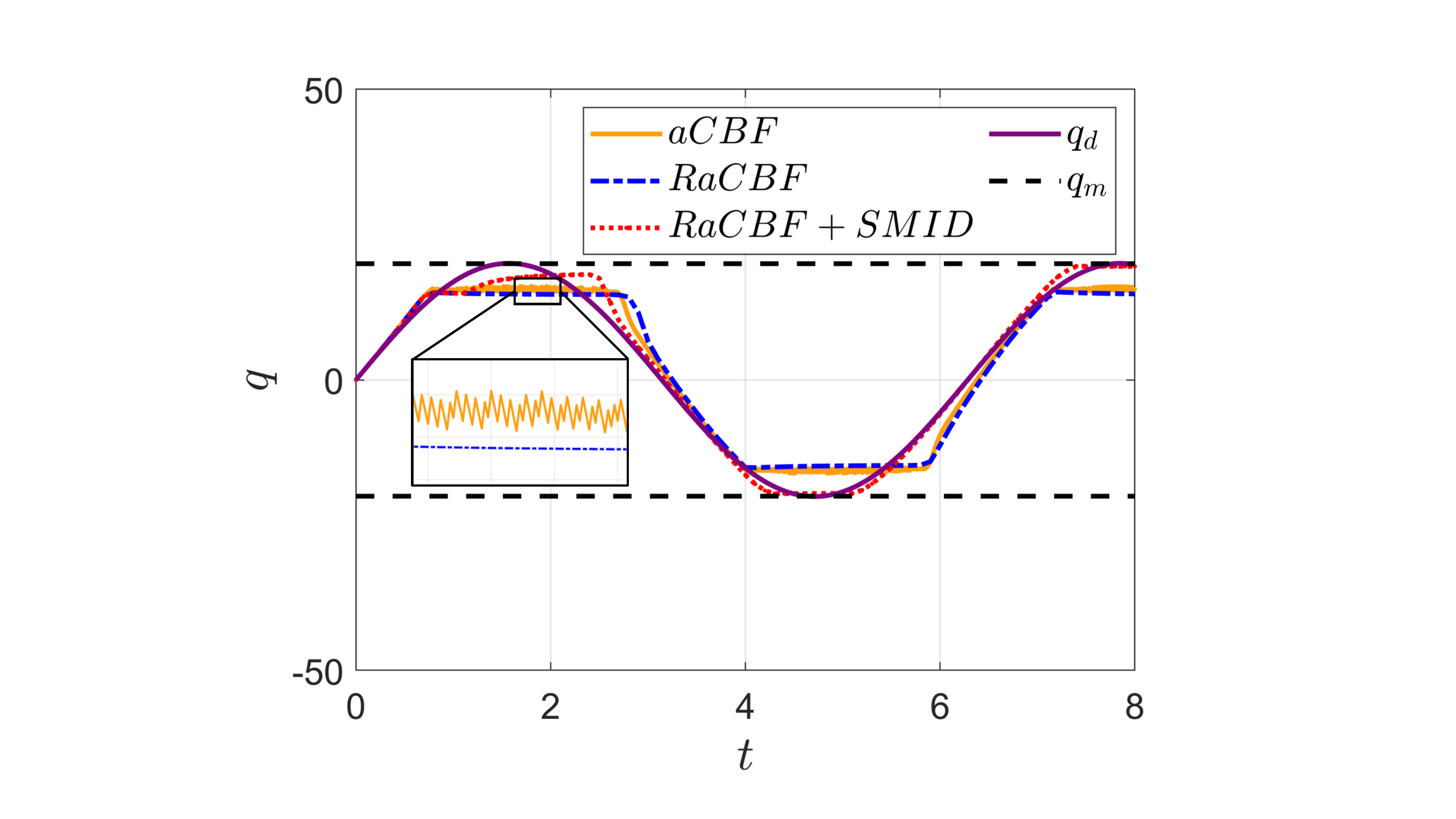}
         \caption{Pitch rate $q$.}
         \label{fig:q}
     \end{subfigure}
     \hfill{}
     \begin{subfigure}{.66\columnwidth}
         \centering
         \includegraphics[trim=150 30 150 30,clip, width=1\textwidth]{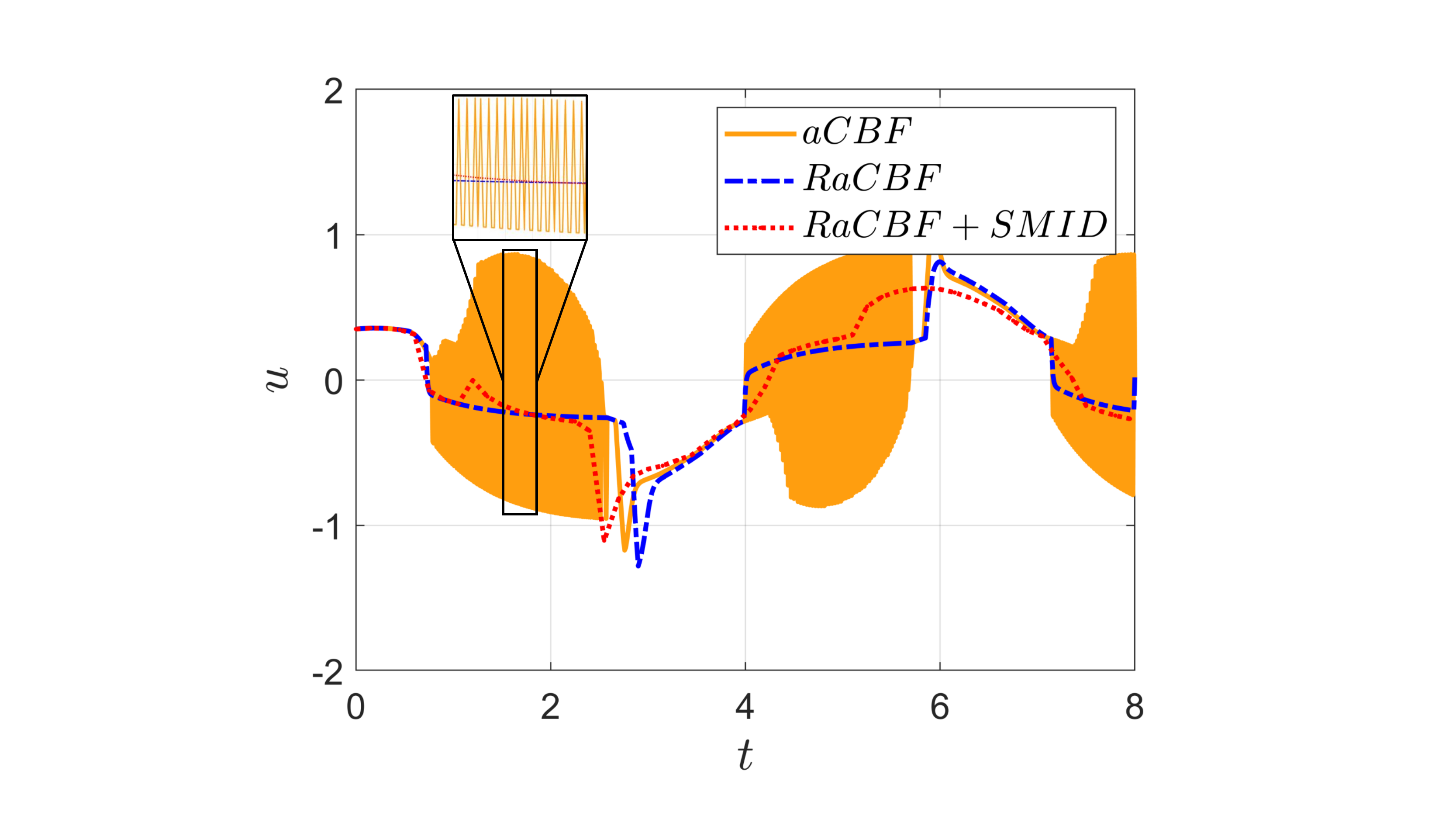}
         \caption{Control input $u$.}
         \label{fig:u}
     \end{subfigure}
     \hfill{}
     \begin{subfigure}{.66\columnwidth}
         \centering
         \includegraphics[trim=150 30 150 30,clip, width=1\textwidth]{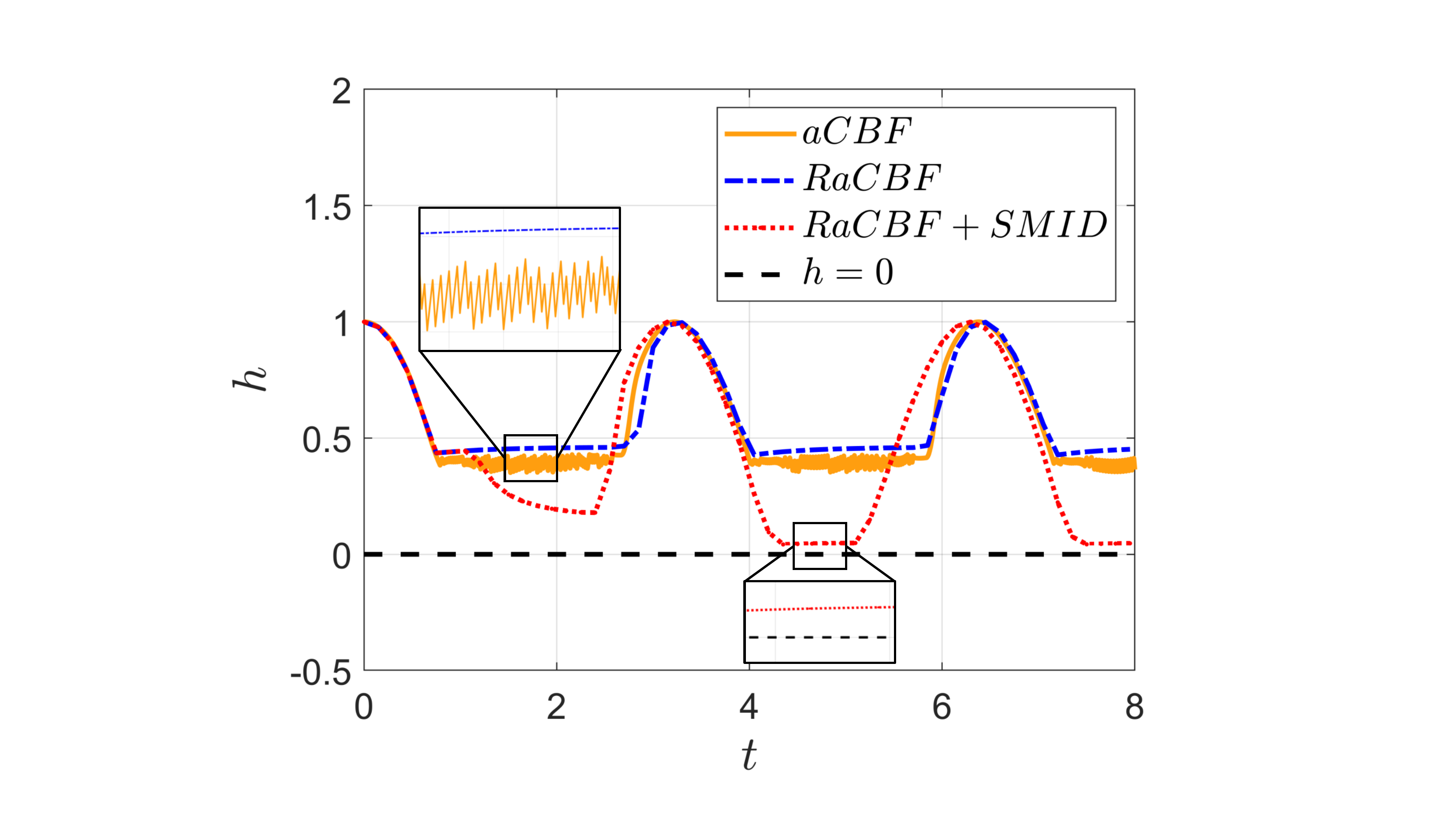}
         \caption{Barrier function $h$.}
         \label{fig:h}
     \end{subfigure}
     \caption{Comparison of modified aCBFs \cref{eq:acbf_m}, RaCBFs, and RaCBFs \& SMID with a full desired trajectory corresponding to $\theta_d = -20^\circ \mathrm{cos}(t)$. (a): Pitch rate tracking where RaCBFs and aCBFs exhibit similar conservativeness due to model error.
     RaCBF \& SMID achieves the best performance because model uncertainty is reduced via online estimation. 
     (b): Control chatter is observed with aCBFs while RaCBFs generate continuous control inputs. (c): Safety maintained but RaCBFs and aCBFs are conservative due to potential model error. RaCBFs \& SMID is the least conservative since model is estimated online.
     For tests $k^* = 0.2$, $\ell^*_\alpha = -1$, $\Gamma_B = 20$, $\Gamma_C = 50$, $\alpha(r) = 10r$, and $D = 0.1$.}
     \label{fig:results}
\vskip -0.2in
\end{figure*}

Now consider the scenario where a full desired trajectory described by $\theta_d = - 20^\circ \mathrm{cos}(t)$ is available.
A metric was synthesized for a new grid range $\alpha \in [-60^{\circ}~60{^\circ}]$ and $q \in [-20^{\circ/s}~20^{\circ/s}]$; a metric quadratic in $\alpha$ was again found to be valid over the grid range.
The function $h_r(q) = 1 - (\nicefrac{q}{q_m})^2$ where $q_m = 20^{\circ/s}$ is a valid RaCBF that enforces $|q| \leq q_m$.
The results in \cref{fig:results} show the same exact behavior as in \cref{fig:results_immelmann}: control input chattering occurs with the modified aCBFs \cref{fig:u} resulting in high-frequency oscillations in both the pitch rate \cref{fig:q} and barrier function \cref{fig:h}. 
Chattering does not occur with RaCBFs.
Additionally, RaCBFs with SMID again has the best tracking performance in \cref{fig:q} and is the least conservative in \cref{fig:h}.
The parameter bounds at different time instances are shown in \cref{fig:smid_sine}.
The bounds again monotonically decrease resulting in a reduction of 16.5\% and 77.3\% for ${k}_q$ and ${\ell}_\alpha$, respectively.
The reduction is less than that in the Immelmann turn due to the trajectory not sufficiently exciting $q$ relative to $\bar{L}(\alpha)$.
The largest reduction occurred at $t=3s$ which is when the barrier function in \cref{fig:h} becomes less conservative.
The computation time shown in \cref{fig:comp_time_sine} is comparable to that in \cref{fig:comp_time_i} with the NLP solve time being slightly less as the linear geodesic initialization is a better initial guess for small tracking error. 
This further confirms that the proposed approach can be run in real-time.

\begin{figure}[t!]
    \begin{subfigure}{.47\columnwidth}
         \centering
         \includegraphics[width=1\linewidth]{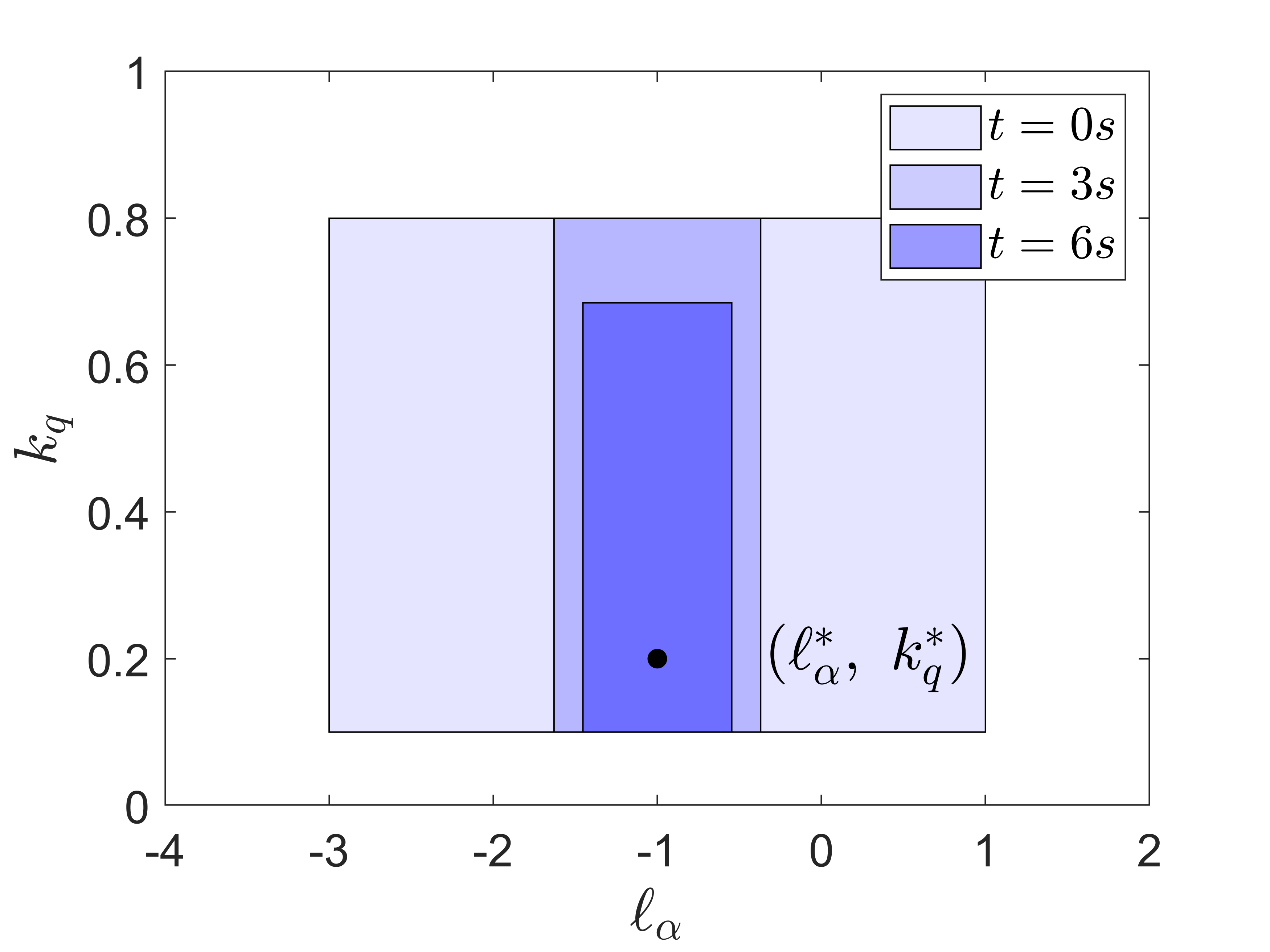}
         \caption{{Parameter bounds.}}
         \label{fig:smid_sine}
     \end{subfigure}
    \hspace{0.6em}
     \begin{subfigure}{.47\columnwidth}
         \centering
         \includegraphics[width=1\linewidth]{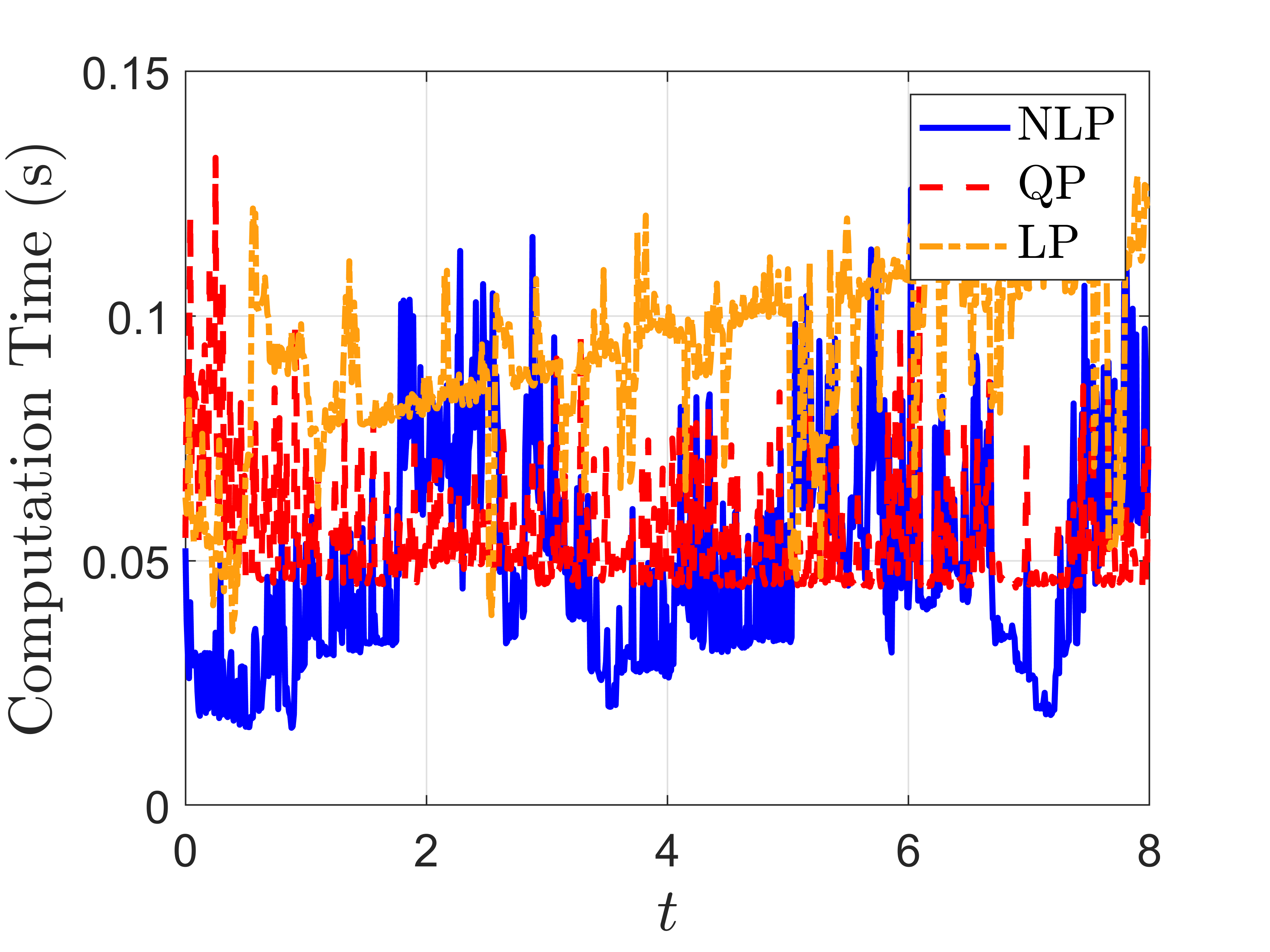}
         \caption{{Computation time.}}
         \label{fig:comp_time_sine}
     \end{subfigure}
    \caption{Parameter bounds and computation time for desired trajectory. (a): Parameter bounds monotonically approach the true parameter values $\ell_\alpha^*$ and $k_q^*$. (b): Computation time for the proposed controller is well within real-time constraints.}
    \label{fig:results_comp_sine}
    \vskip -0.2in
\end{figure}

\section{Conclusion}
This work presented a framework that guarantees safety for uncertain nonlinear systems through parameter adaptation and data-driven model estimation.
The unification with a contraction-based adaptive controller allows the approach be applied to a broad class of systems.
Extending to systems with probabilistic model bounds, non-parametric uncertainties, and external disturbances is future work.

\section*{Appendix}
\label{sec:appendix}

\begin{proof}[Proof of \cref{thm:racbf}]
Consider the composite candidate CBF $h = h_r(x,\hat{\theta}) - \frac{1}{2} \tilde{\theta}^\top\Gamma^{-1}\tilde{\theta}$,
where the minimum eigenvalue of $\Gamma$ must satisfy $\lambda_{\min}(\Gamma) \geq \frac{\|\tilde{\vartheta}\|^2}{2h_r(x_r,{\theta}_r)}$ for any $h_r(x_r,\theta_r)>0$.
Differentiating $h$ with respect to \cref{eq:unc_dyn},
\begin{equation*}
\begin{aligned}
    \dot{h} & = \dot{h}_r(x,\hat{\theta}) -  \tilde{\theta}^\top\Gamma^{-1}\dot{\hat{\theta}} \\
    &= \frac{\partial h_r}{\partial x}\left[ f(x) - \Delta(x)^\top \theta + B(x) u\right] + \frac{\partial h_r}{\partial \hat{\theta}} \dot{\hat{\theta}} - \tilde{\theta}^\top\Gamma^{-1}\dot{\hat{\theta}}  \\
\end{aligned}
\end{equation*}
Adding and subtracting $\frac{\partial h_r}{\partial x}\Delta(x)^\top\left[\hat{\theta} - \Gamma \left(\frac{\partial h_r}{\partial \hat{\theta}}\right)^\top \right]$ and using the definition of $\Lambda(x,\hat{\theta})$,

\begin{equation*}
\begin{aligned}
    \dot{h} & = \frac{\partial h_r}{\partial x}\left[ f(x) - \Delta(x)^\top \Lambda(x,\hat{\theta}) + B(x) u\right] + \frac{\partial h_r}{\partial \hat{\theta}}\dot{\hat{\theta}} \\
    & \hspace{1.5cm} - \tilde{\theta}^\top\Gamma^{-1}\dot{\hat{\theta}} + \frac{\partial h_r}{\partial x}\Delta(x)^\top\left[\tilde{\theta} - \Gamma \left(\frac{\partial h_r}{\partial \hat{\theta}}\right)^\top \right].
\end{aligned}
\end{equation*}
Choosing $\dot{\hat{\theta}} = \Gamma \Delta(x) \left(\frac{\partial h_r}{\partial x}\right)^\top$, then
\begin{equation*}
\begin{aligned}
    \dot{h} &= \frac{\partial h_r}{\partial x}\left[ f(x) - \Delta(x)^\top \Lambda(x,\hat{\theta}) + B(x) u\right] \\
    & \geq - \alpha\left( h_r - \frac{1}{2}\tilde{\vartheta}^\top \Gamma^{-1} \tilde{\vartheta}  \right) \geq -\alpha(h),
\end{aligned}
\end{equation*}
where the first inequality is obtained via the definition of a RaCBF and the second by noting $|\tilde{\theta}| \leq \tilde{\vartheta}$ so $h = h_r - \frac{1}{2}\tilde{\theta^\top}\Gamma^{-1}\tilde{\theta} \geq h_r - \frac{1}{2}\tilde{\vartheta^\top}\Gamma^{-1}\tilde{\vartheta}$. 
Since $h\geq0$ and $h_r \geq h$ $\forall t$, then $h_r \geq \frac{1}{2}\tilde{\vartheta^\top}\Gamma^{-1}\tilde{\vartheta}  \geq 0$ and $\mathcal{C}^r_\theta$ is forward invariant.
\end{proof}

\begin{proof}[Proof of \cref{thm:racbf_smid}]
Since the model uncertainty is changing via estimation, the maximum allowable parameter error is time varying, i.e., $\tilde{\vartheta}(t)$.
From \cref{lemma:smid}, $\tilde{\Theta}$ monotonically decreases so $\dot{\tilde{\vartheta}} \leq 0$. 
Let $h_r$ be a candidate RaCBF, then $\dot{{h}} = \dot{h}_{r} - \tilde{\vartheta}\Gamma^{-1}\dot{\tilde{\vartheta}}  \geq \dot{h}_{r}$ since $\dot{\tilde{\vartheta}} \leq 0$ for all $t$.
Inequality \cref{eq:racbf} in Definition~\ref{definition:racbf} is then still satisfied for $\dot{\tilde{\vartheta}} \leq 0$.
Using the steps in \cref{thm:racbf}, the system is safe with respect to $\mathcal{C}^r_{\hat{\theta}}$.
Moreover, since $\dot{\tilde{\vartheta}} \leq 0$ then $\tilde{\vartheta} \rightarrow 0 $ so $\mathcal{C}^r_{\hat{\theta}} \rightarrow \mathcal{C}_{\hat{\theta}}$.
\end{proof}


\noindent \textbf{Acknowledgements}  
We thank David Fan for stimulating discussions.
This work was supported by the NSF Graduate Research Fellowship Grant No. 1122374.


\balance
\bibliographystyle{ieeetr}
\bibliography{ref}

\end{document}